\documentclass[aps,prl,groupedaddress,superscriptaddress,twocolumn,notitlepage]{revtex4-1}

\usepackage[latin1]{inputenc}
\usepackage{amsthm}
\usepackage{amssymb}
\usepackage{amsmath}
\usepackage{bbold}
\usepackage{bbm}
\usepackage{nonfloat}
\usepackage[pdftex]{hyperref}
\usepackage{braket}
\usepackage{dsfont}
\usepackage{mathdots}
\usepackage{mathtools}
\usepackage{enumerate}
\usepackage[shortlabels]{enumitem}
\usepackage{csquotes}
\usepackage{stmaryrd}
\usepackage[cal=boondox]{mathalfa}
\usepackage{graphicx}
\usepackage{stackengine}
\usepackage{scalerel}
\usepackage{array}
\usepackage{makecell}
\newcolumntype{x}[1]{>{\centering\arraybackslash}p{#1}}
\usepackage{tikz}
\usepackage{pgfplots}
\usetikzlibrary{shapes.geometric, shapes.misc, positioning, arrows, arrows.meta, decorations.pathreplacing, decorations.pathmorphing, patterns, angles, quotes, calc}
\usepackage{booktabs}
\usepackage{xfrac}
\usepackage{siunitx}
\usepackage{centernot}
\usepackage{comment}
\usepackage{chngcntr}
\usepackage{caption}
\usepackage{subcaption}

\newtheorem{thm}{Theorem}
\newtheorem*{thm*}{Theorem}
\newtheorem{prop}[thm]{Proposition}
\newtheorem*{prop*}{Proposition}
\newtheorem{lemma}[thm]{Lemma}
\newtheorem*{lemma*}{Lemma}
\newtheorem{cor}[thm]{Corollary}
\newtheorem*{cor*}{Corollary}

\newtheorem*{cj*}{Conjecture}
\newtheorem{Def}[thm]{Definition}
\newtheorem*{Def*}{Definition}

\newtheorem*{question*}{Question}

\newtheorem*{problem*}{Problem}

\makeatletter
\def\thmhead@plain#1#2#3{%
  \thmname{#1}\thmnumber{\@ifnotempty{#1}{ }\@upn{#2}}%
  \thmnote{ {\the\thm@notefont#3}}}
\let\thmhead\thmhead@plain
\makeatother

\theoremstyle{definition}
\newtheorem{rem}[thm]{Remark}
\newtheorem*{note}{Note}

\newcommand{\bb}{\begin{equation}\begin{aligned}\hspace{0pt}}
\newcommand{\bbb}{\begin{equation*}\begin{aligned}}
\newcommand{\ee}{\end{aligned}\end{equation}}
\newcommand{\eee}{\end{aligned}\end{equation*}}
\newcommand*{\coloneqq}{\mathrel{\vcenter{\baselineskip0.5ex \lineskiplimit0pt \hbox{\scriptsize.}\hbox{\scriptsize.}}} =}
\newcommand*{\eqqcolon}{= \mathrel{\vcenter{\baselineskip0.5ex \lineskiplimit0pt \hbox{\scriptsize.}\hbox{\scriptsize.}}}}
\newcommand\floor[1]{\left\lfloor#1\right\rfloor}
\newcommand\ceil[1]{\left\lceil#1\right\rceil}
\newcommand{\texteq}[1]{\stackrel{\mathclap{\scriptsize \mbox{#1}}}{=}}
\newcommand{\textleq}[1]{\stackrel{\mathclap{\scriptsize \mbox{#1}}}{\leq}}

\newcommand{\textgeq}[1]{\stackrel{\mathclap{\scriptsize \mbox{#1}}}{\geq}}
\newcommand{\ketbra}[1]{\ket{#1}\!\!\bra{#1}}

\newcommand{\sumno}{\sum\nolimits}

\newcommand{\tcr}[1]{{\color{red} #1}}

\newcommand{\id}{\mathds{1}}
\newcommand{\R}{\mathds{R}}
\newcommand{\N}{\mathds{N}}

\newcommand{\C}{\mathds{C}}

\newcommand{\locc}{\mathrm{LOCC}}

\newcommand{\gocc}{\mathrm{GOCC}}
\newcommand{\het}{\mathrm{het}}
\renewcommand{\hom}{\mathrm{hom}}

\DeclareMathOperator{\Tr}{Tr}

\DeclareMathAlphabet{\pazocal}{OMS}{zplm}{m}{n}

\DeclareMathOperator{\erf}{erf}
\DeclareMathOperator{\dom}{dom}

\newcommand{\HH}{\pazocal{H}}
\newcommand{\T}{\pazocal{T}}

\newcommand{\B}{\pazocal{B}}

\newcommand{\NN}{\mathcal{N}}
\newcommand{\MM}{\mathcal{M}}

\newcommand{\lsmatrix}{\left(\begin{smallmatrix}}
\newcommand{\rsmatrix}{\end{smallmatrix}\right)}

\stackMath

\stackMath

\makeatletter
\newcommand*\rel@kern[1]{\kern#1\dimexpr\macc@kerna}
\newcommand*\widebar[1]{%
  \begingroup
  \def\mathaccent##1##2{%
    \rel@kern{0.8}%
    \overline{\rel@kern{-0.8}\macc@nucleus\rel@kern{0.2}}%
    \rel@kern{-0.2}%
  }%
  \macc@depth\@ne
  \let\math@bgroup\@empty \let\math@egroup\macc@set@skewchar
  \mathsurround\z@ \frozen@everymath{\mathgroup\macc@group\relax}%
  \macc@set@skewchar\relax
  \let\mathaccentV\macc@nested@a
  \macc@nested@a\relax111{#1}%
  \endgroup
}

\newcommand{\fakepart}[1]{
 \par\refstepcounter{part}
  \sectionmark{#1}
}

\counterwithin*{equation}{part}
\counterwithin*{thm}{part}
\counterwithin*{figure}{part}

\tikzset{meter/.append style={draw, inner sep=10, rectangle, font=\vphantom{A}, minimum width=30, line width=.8, path picture={\draw[black] ([shift={(.1,.3)}]path picture bounding box.south west) to[bend left=50] ([shift={(-.1,.3)}]path picture bounding box.south east);\draw[black,-latex] ([shift={(0,.1)}]path picture bounding box.south) -- ([shift={(.3,-.1)}]path picture bounding box.north);}}}
\tikzset{roundnode/.append style={circle, draw=black, fill=gray!20, thick, minimum size=10mm}}
\tikzset{squarenode/.style={rectangle, draw=black, fill=none, thick, minimum size=10mm}}

\definecolor{Blues5seq1}{RGB}{239,243,255}
\definecolor{Blues5seq2}{RGB}{189,215,231}
\definecolor{Blues5seq3}{RGB}{107,174,214}
\definecolor{Blues5seq4}{RGB}{49,130,189}
\definecolor{Blues5seq5}{RGB}{8,81,156}

\definecolor{Greens5seq1}{RGB}{237,248,233}
\definecolor{Greens5seq2}{RGB}{186,228,179}
\definecolor{Greens5seq3}{RGB}{116,196,118}
\definecolor{Greens5seq4}{RGB}{49,163,84}
\definecolor{Greens5seq5}{RGB}{0,109,44}

\definecolor{Reds5seq1}{RGB}{254,229,217}
\definecolor{Reds5seq2}{RGB}{252,174,145}
\definecolor{Reds5seq3}{RGB}{251,106,74}
\definecolor{Reds5seq4}{RGB}{222,45,38}
\definecolor{Reds5seq5}{RGB}{165,15,21}

\newcommand{\DD}{\mathcal{D}}
\newcommand{\EE}{\pazocal{E}}
\newcommand{\Wp}{{\pazocal{W}_+}}

\newcommand{\homd}[2]{\left\|#1-#2\right\|_{\mathrm{hom}}}
\newcommand{\hetd}[2]{\left\|#1-#2\right\|_{\mathrm{het}}}

\usepackage{pgfplots}
\pgfplotsset{width=10cm,compat=1.9}
\usetikzlibrary{decorations.markings}

\renewcommand{\tcr}{}

\begin{document}

\title{Quantum data hiding with continuous variable systems}

\author{Ludovico Lami}
\email{ludovico.lami@gmail.com}
\affiliation{Institut f\"{u}r Theoretische Physik und IQST, Universit\"{a}t Ulm, Albert-Einstein-Allee 11, D-89069 Ulm, Germany}

\begin{abstract}
Suppose we want to benchmark a quantum device held by a remote party, e.g.\ by testing its ability to carry out challenging quantum measurements outside of a free set of measurements $\mathcal{M}$. A very simple way to do so is to set up a binary state discrimination task that cannot be solved efficiently by means of free measurements. If one can find pairs of orthogonal states that become arbitrarily indistinguishable under measurements in $\mathcal{M}$, in the sense that the error probability in discrimination approaches that of a random guess, one says that there is data hiding against $\mathcal{M}$. Here we investigate data hiding in the context of continuous variable quantum systems. First, we look at the case where $\mathcal{M}=\mathrm{LOCC}$, the set of measurements implementable with local operations and classical communication. While previous studies have placed upper bounds on the maximum efficiency of data hiding in terms of the local dimension and are thus not applicable to continuous variable systems, we establish more general bounds that rely solely on the local mean photon number of the states employed. Along the way, we perform a rigorous quantitative analysis of the error introduced by the non-ideal Braunstein--Kimble quantum teleportation protocol, determining how much squeezing and local detection efficiency is needed in order to teleport an arbitrary multi-mode local state of known mean energy with a prescribed accuracy. Finally, following a seminal proposal by Sabapathy and Winter, we look at data hiding against Gaussian operations assisted by feed-forward of measurement outcomes, providing the first example of a relatively simple scheme that works with a single mode only.
\end{abstract}

\maketitle
\fakepart{Main text}

\section{Introduction} \label{intro_sec}
Quantum state discrimination, which consists in identifying by means of a measurement an unknown quantum state, is one of the most fundamental tasks studied by quantum information science. Indeed, it captures the essential nature of scientific hypothesis testing in the quantum realm. At the same time, its conceptual simplicity allows for an in-depth mathematical analysis.
The systematic study of the problem was initiated by Holevo~\cite{Holevo1972-analogue, Holevo1973-statistical, Holevo1976} and Helstrom~\cite{Helstrom-paper, HELSTROM}, who solved the binary case, featuring only two states $\rho$ and $\sigma$: the lowest error probability in discrimination is given by
\bb
P_e(\rho,\sigma; p) = \frac12 \left( 1-\left\|p\rho-(1-p)\sigma\right\|_1\right) ,
\label{Helstrom}
\ee
where $\|X\|_1 \coloneqq \Tr |X|$ is the trace norm, and $p,1-p$ are the a priori probabilities of $\rho,\sigma$, respectively. The investigation of quantum state discrimination in the asymptotic Stein setting has revealed operational interpretations for the quantum relative entropy~\cite{Umegaki1962, Hiai1991, Ogawa2000} and for the regularised relative entropy of entanglement~\cite{Vedral1997, Plenio2000, Brandao2010, BrandaoPlenio1, BrandaoPlenio2}. Also the symmetric Chernoff setting has been thoroughly studied~\cite{Nussbaum2009, qChernoff, Audenaert2008}. Recently, quantum \emph{channel} discrimination~\cite{Aharonov1998, WATROUS} has also become an important research topic~\cite{Acin2001, Duan2007, Duan2008, Duan2009}, and its investigation either in the asymptotic~\cite{Berta2018, Chiribella2019, Fang2020, Salek2020} or in the energy-constrained~\cite{Shirokov2018, VV-diamond, Berta2018, EC-diamond} setting has just witnessed notable progress.

A common assumption that is made when designing discrimination protocols is that any quantum measurement can in principle be carried out. This assumption, however, does not take into account technological hurdles, which are likely to be a major concern on near-term quantum devices~\cite{Preskill2018, Bharti2021}, nor more fundamental constraints that appear e.g.\ in multi-partite settings, where only local quantum operations assisted by classical communication (LOCC) are available~\cite{LOCC}.
In general, let $\MM$ be the class of available quantum measurements. Following the notation of~\eqref{Helstrom}, the lowest probability of error in discriminating $\rho$ and $\sigma$ by employing measurements from $\MM$ is denoted by $P_e^{\MM}$, where we omitted the dependence from the `scheme' $(\rho, \sigma; p)$.

For certain sets $\MM$, a phenomenon called \emph{quantum data hiding} against $\MM$ can emerge~\cite{dh-original-1, dh-original-2, hiding-data, randomising, VV-dh-Chernoff, VV-dh, Chitambar2014, Aubrun-2015, ultimate, XOR, Cheng2020}. Namely, it is possible for two states to be perfectly distinguishable when all measurements are available, yet to be almost indistinguishable when only measurements from $\MM$ are employed. In formula, this means that $P_e \approx 0$ while $P_e^{\MM} \approx 1/2$ for a certain scheme $(\rho,\sigma;p)$. 

This genuinely quantum behaviour can be exploited, for instance, to design a simple protocol to benchmark a quantum device held by a remote party. This protocol is based on (repetitions of) a state discrimination task $(\rho, \sigma; p)$: the unknown state is sent to the remote party, which is tasked with identifying it. If $\rho$ and $\sigma$ are almost indistinguishable under measurements in $\MM$, a correct solution to this problem invariably indicates that the remote device is able to carry out measurements that lie outside of $\MM$.

The original construction provides a striking example of data hiding against the set of LOCC-implementable measurements~\cite{dh-original-1, dh-original-2}. Discrimination of quantum states or channels in the presence of locality constraints has been studied in several different contexts~\cite{no-dh-pure-1, no-dh-pure-2, Eggeling2002, Duan2008, state-discr-GPT, Matthews2010, Brandao-area-law, Lupo2016, Arai2019, XOR}, but always with the assumption of finite dimension. However, many systems of physical and technological interest are intrinsically infinite dimensional. Among these, continuous variable (CV) quantum systems~\cite{CERF, HOLEVO, BUCCO, HOLEVO-CHANNELS-2, Braunstein-review, weedbrook12}, which model e.g.\ electromagnetic modes travelling along an optical fibre, will likely play a major role in the future of quantum technologies, with applications ranging from quantum computation~\cite{KLM, OBrien2009, Politi2009, Carolan2015, Rohde2015, Kumar2020} to quantum communication~\cite{Flamini2018}.

In the case of a CV quantum system, another restricted set of measurements --- identified by technological constraints --- comprises all Gaussian operations assisted by classical computation, or feed-forward of measurement outcomes (GOCC)~\cite{Takeoka-Sasaki}. (Despite the similar name, the set of GOCC operations has nothing to do with that of LOCC.) The theoretical study of state discrimination with GOCC has been pioneered by Takeoka and Sasaki~\cite{Takeoka-Sasaki} and continued by Sabapathy and Winter, who have provided the first example of quantum data hiding against GOCC, achieved through a randomised construction that requires asymptotically many optical modes~\cite{KK-VV-GOCC, VV-GOCC-IHP, VV-GOCC-APS, VV-GOCC-ISIT}. Let us remark in passing that it is only recently that experiments have been able to achieve a better performance than GOCC at coherent state discrimination~\cite{Tsujino2011, Becerra2013, Sych2016, Rosati2016, DiMario2018, Mueller-experimental-1, Mueller-experimental-2, Ferdinand2017, Sidhu2021}.

In this paper, we investigate the phenomenon of data hiding in the CV setting, thus filling a striking gap in the existing literature. First, for a bipartite system we show how to upper bound the maximum efficiency of data hiding against LOCC in terms of the local energy (i.e.\ mean photon number) of the employed states. 
Our result implies that high-efficiency data hiding necessarily requires high-energy states. 
This behaviour is analogous to that of finite-dimensional schemes~\cite{ultimate, VV-dh}, with the energy playing the role of an effective dimension.

Our technique relies on an explicit bound on the maximum disturbance introduced by the Braunstein--Kimble CV quantum teleportation protocol~\cite{Vaidman1994, Braunstein1998}, in terms of the squeezing of the resource state, of the efficiency of the local detectors, and of the mean energy of the input state (Theorem~\ref{BK_accuracy_thm}). With such a bound one is able to assess the resources needed to teleport a state of a given mean energy with a certain maximum allowed error, as measured by the operationally meaningful trace norm.

The second part of our paper deals with data hiding against GOCC. 
We construct an explicit single-mode scheme that achieves data hiding against GOCC with arbitrarily high efficiency, thus answering a question raised by Sabapathy and Winter~\cite{KK-VV-GOCC, VV-GOCC-APS}. The states we employ can be obtained from two-mode squeezed vacua by performing a destructive observation of the photon number parity on one mode, and retaining the reduced state on the other.
We also analyse the performance of other schemes, involving either thermal or Fock states, and prove that they do not exhibit data hiding against GOCC. This illustrates the versatility of GOCC measurements for state discrimination in many practical settings, and at the same time it this singles out our example as a candidate for a relatively simple protocol that demonstrates data hiding with a single mode only.


\section{Preliminaries and notation}

Quantum systems, hereafter denoted with letters such as $A,B$, and so on, are mathematically described by Hilbert spaces. A bounded operator on a Hilbert space $\HH$ is a linear map $T:\HH \to \HH$ such that $\sup_{\ket{x}\in \HH\setminus 0}\frac{\left\|T\ket{x}\right\|}{\|\ket{x}}\| \eqqcolon \|T\|_\infty <\infty$. Bounded operators on $\HH$ form a Banach space, denoted with $\B (\HH)$, once they are equipped with the operator norm $\|\cdot\|_\infty$. A bounded operator $T$ such that $\Tr|T| = \Tr \sqrt{T^\dag T} \eqqcolon \|T\|_1 <\infty$ is said to be of trace class. Trace class operators on $\HH$ form another Banach space, denoted with $\T(\HH)$, once they are equipped with the trace norm $\| \cdot \|_1$. Quantum states on $\HH$ are represented by density operators, that is, positive semi-definite trace class operators with trace one. All quantum states $\rho$ admit a spectral decomposition $\rho = \sum_i p_i \ketbra{\psi_i}$, where the vectors $\ket{\psi_i}$ are orthonormal, $p_i \geq 0$ for all $i$, and $\sum_i p_i=1$~\cite{HALL}. Composition between quantum systems is represented by tensor products: if $A$ is described by the Hilbert space $\HH_A$ and $B$ by the Hilbert space $\HH_B$, the bipartite system $AB$ will be described by the Hilbert space $\HH_{AB}\coloneqq \HH_A\otimes \HH_B$. Given a quantum state $\rho$ decomposed as before and a generic positive semi-definite (not necessarily bounded) operator $L\geq 0$ with domain $\dom(L)\subseteq \HH$, the expectation value of $L$ on $\rho$ is defined by~\cite{SCHMUEDGEN}
\bb
\Tr[\rho L] \coloneqq&\ \left\{ \begin{array}{ll} \sum_i p_i \left\|L^{1/2} \ket{\psi_i}\right\|^2 & \text{if $\ket{\psi_i}\in \dom(L)\ \forall\, i$,} \\[1.5ex] +\infty & \text{otherwise.} \end{array} \right.
\ee

A generic measurement on a quantum system with Hilbert space $\HH$ is represented by a normalised \emph{positive operator-valued measure} (POVM) $E(dx)$ over a measurable space $\pazocal{X}$, which models the set of possible outcomes~\cite[Definition~11.29]{HOLEVO-CHANNELS-2}. More formally, we consider a $\sigma$-algebra $\pazocal{A}$ on $\pazocal{X}$, and introduce $E$ as a function $E:\pazocal{A}\to \B_+(\HH)$ taking on values in the space of positive semi-definite bounded operators on $\HH$. The normalisation condition amounts to
\bb
\int_\pazocal{X} E(dx) = \id\, .
\label{normalisation}
\ee
The measurement $E$ applied to the state $\rho$ yields as outcome a random variable $X$ distributed with measure
\bb
\mu_X(dx) = \Tr\left[ \rho\, E(dx) \right] .
\label{outcome}
\ee

A quantum channel with input system $A$ and output system $B$ is represented (in Schr\"odinger picture) by a completely positive trace-preserving (CPTP) map $\Phi:\T (\HH_A)\to \T (\HH_B)$. The complementary Heinsenberg representation of the action of $\Phi$ is obtained by taking its adjoint, i.e.\ by constructing the map $\Phi^\dag:\B (\HH_B)\to \B (\HH_A)$ defined by
\bb
\Tr\left[ \left(\Phi^\dag (X)\right)^\dag Y\right] \equiv \Tr\left[ X^\dag \Phi(Y)\right] ,
\label{adjoint}
\ee
for all $X\in \B (\HH_B)$ and $Y\in \T (\HH_A)$.

\subsection{Continuous variable systems} \label{CV_systems_subsec}

The Hilbert space corresponding to an $m$-mode CV system is given by $\HH_m \coloneqq L^2(\R^m)$. It comprises all square-integrable complex-valued functions over the Euclidean space $\R^m$. The canonical operators $x_j$ and $p_j\coloneqq -i\frac{\partial}{\partial x_j}$ satisfy the \emph{canonical commutation relations}
\bb
[x_j, x_k] \equiv 0 \equiv [p_j, p_k]\, ,\qquad [x_j, p_k] = i \delta_{jk} I\, , 
\label{CCR_xp}
\ee
which can be recast in the form
\bb
[a_j, a_k] \equiv 0\, ,\qquad [a_j, a_k^\dag] = \delta_{jk} I
\label{CCR}
\ee
in terms of the \emph{annihilation and creation operators}
\bb
a_j \coloneqq \frac{x_j + i p_j}{\sqrt{2}}\, ,\qquad a_j \coloneqq \frac{x_j - i p_j}{\sqrt{2}}\, .
\label{a_adag}
\ee
Here, $I$ denotes the identity operator over the Hilbert space $\HH_m$. Creation operators transform the \emph{vacuum state} $\ket{0}$ into the \emph{Fock states}
\bb
\ket{n} = \bigotimes_{j=1}^m \ket{n_j} \coloneqq \left(\bigotimes\nolimits_{j=1}^m \frac{(a_j^\dag)^{n_j}}{\sqrt{n_j!}} \right) \ket{0}\, ,
\label{Fock}
\ee
where $n\coloneqq (n_1, \ldots, n_m)\in \N^m$ is a tuple of non-negative integers. Fock states are eigenvectors of the \emph{total photon number} Hamiltonian
\bb
N\coloneqq \sum_{j=1}^m a_j^\dag a_j\, ,
\label{total_photon_number}
\ee
which satisfies
\bb
N\ket{n} = \left( \sumno_j n_j \right) \ket{n}\, .
\ee

A central role in the theory is played by the \emph{displacement operators}, defined for $\alpha\in \C^m$ by~\cite[Eq.~(3.3.30)]{BARNETT-RADMORE}
\bb
D(\alpha) \coloneqq \exp \left[ \sumno_{j=1}^m \left( \alpha_j a_j^\dag - \alpha_j^* a_j \right) \right] .
\label{displacement}
\ee
Note that $D(\alpha)$ is unitary, and that $D(\alpha)^\dag = D(-\alpha)$. The canonical commutation relations~\eqref{CCR_xp}--\eqref{CCR} can be expressed in their \emph{Weyl form} as
\bb
D(\alpha) D(\beta) = D(\alpha+\beta)\, e^{\frac12 (\alpha^\intercal \beta^* - \alpha^\dag \beta)}\, .
\label{CCR_Weyl}
\ee
Displacement operators derive their name from the fact that they effectively act on canonical operators by translating them, i.e.
\bb
D(\alpha)^\dag\, a_j\, D(\alpha) = a_j + \alpha_j\, .
\label{displacement_annihilation}
\ee

By taking the trace of an arbitrary trace class operator $T\in \T(\HH_m)$ against displacement operators --- which are unitary and thus bounded --- we can construct the \emph{characteristic function} $\chi_T:\C^m\to \C$, defined by
\bb
\chi_T(\alpha) \coloneqq \Tr\left[ T D(\alpha) \right] \, .
\label{chi}
\ee
Its Fourier transform is known as the \emph{Wigner function}~\cite{Wigner, Grossmann1976, Hillery1984}. It can be expressed in many different ways, namely~\cite[Eq.~(4.5.12),~(4.5.19), and~(4.5.21)]{BARNETT-RADMORE}
\begin{align}
W_T(\alpha) \coloneqq&\ \int \frac{d^{2m}\beta}{\pi^{2m}}\, \chi_T(\beta)\, e^{\alpha^\intercal \beta^* - \alpha^\dag \beta} \label{Wigner_1}\\[0.5ex]
=&\ \frac{2^m}{\pi^m} \Tr\left[ D(-\alpha)T D(\alpha)\, (-1)^{\sum_j a_j^\dag a_j}\right] \label{Wigner_2}
\\[0.5ex]
=&\ \frac{2^m}{\pi^m} \Tr\left[ T\, D(2\alpha)\, (-1)^{\sum_j a_j^\dag a_j}\right] . \label{Wigner_3}
\end{align}
Note that
\bb
\int d^{2m} \alpha\, W_T(\alpha) = \chi_T(0) = \Tr T\, .
\label{integrate_Wigner}
\ee
States with a Gaussian Wigner function are called \emph{Gaussian states}~\cite{weedbrook12, BUCCO}. \emph{Gaussian unitaries} are by definition those unitary operators that can be written as products of factors of the form $e^{iH_q}$, where
\bb
H_q &= \sum_j (\alpha_j a_j^\dag - \alpha_j^* a_j) \\
&\quad + \sum_{j,k} \left(X_{jk} a_j^\dag a_k + Y_{jk} a_j a_k +  Y_{jk}^* a_j^\dag a_k^\dag \right)
\ee
is a generic quadratic expression in the creation and annihilation operators. Here, $X=X^\dag$ is an $m\times m$ Hermitian matrix, and $Y=Y^\intercal$ is an $m\times m$ complex symmetric matrix. Note for instance that the displacement operators~\eqref{displacement} are Gaussian unitaries. More generally, in many systems of experimental interest the natural Hamiltonian is already quadratic in the canonical operators; for this reason, Gaussian unitaries are easier to implement than general unitaries.

As for the measurements, the archetypical measurement on a CV system is the \emph{homodyne detection}, i.e.\ the measurement of a canonical quadrature such as $x_j$, $p_k$, or more generally $\sum_j (a_j x_j + b_j p_j)$, where $a_j\in \R$ and $b_j\in \R$ are arbitrary real coefficients. Homodyne detections are routinely implemented in experimental practice. Thus, a natural question to ask is what can be done by combining them with Gaussian unitaries, which are also relatively inexpensive to perform. Following Takeoka and Sasaki's terminology~\cite{Takeoka-Sasaki}, we call the measurements achievable in this way \emph{GOCC measurements}. They are all those that can be implemented by (i)~adding ancillary modes in the vacuum, (ii)~applying Gaussian unitaries, and (iii)~making a homodyne detection on some mode. It is understood that these operations can be applied sequentially, and that each of them can depend on previous measurement outcomes. For a more formal definition, we refer the reader to Ref.~\cite[SM, Definition~S6]{GIE-LL}.

Apart from the homodyne, another example of a GOCC measurement is the \emph{heterodyne}, given by the POVM $E_{\mathrm{het}}(d^{2m}\alpha) = \ketbra{\alpha}\, \frac{d^{2m}\alpha}{\pi^m}$ on $\C^m$. Here, $\ket{\alpha}\coloneqq D(\alpha)\ket{0}$ is a \emph{coherent state}~\cite{Schroedinger1926-coherent, Klauder1960, Glauber1963, Sudarshan1963}. When a heterodyne is carried out on $\rho$, the resulting probability distribution 
\bb
Q_\rho(\alpha) \coloneqq \frac{1}{\pi^m}\braket{\alpha|\rho|\alpha}
\label{Husimi}
\ee
is referred to as its \emph{Husimi function}~\cite{Husimi}. To see that the heterodyne is also a GOCC measurement, it suffices to observe that it can be implemented by mixing the input system with a vacuum state in a 50:50 beam splitter and separately homodyning the output modes~\cite[p.~122]{BUCCO}.

\section{On the accuracy of the Braunstein--Kimble teleportation protocol}

\tcr{A fundamental technical tool in our forthcoming analysis of data hiding against local operations and classical communication (LOCC) in CV systems is the Braunstein--Kimble teleportation protocol~\cite{Vaidman1994, Braunstein1998, teleportation-review}. To understand why this such a key tool, note that the best upper bound on the strength of data hiding against LOCC in finite-dimensional systems comes from~\cite[Theorem~16]{ultimate}, and the crux of the proof of that result rests precisely on the teleportation protocol. Accordingly, we may expect that an analogous bound on the strength of LOCC data hiding in CV systems may come from a quantitative understanding of CV teleportation. This is indeed the case: the proof of Theorem~\ref{telep_multimode_thm}, which we will obtain in Section~\ref{locc_dh_sec}, depends crucially on Theorem~\ref{BK_accuracy_thm}, which we will prove at the end of this section. On the other hand, CV teleportation is such an important protocol that we chose to separate its analysis, which is independent from data hiding or any of the related questions, from the rest of the paper.}

The Braunstein--Kimble teleportation protocol allows to teleport an $m$-mode system $A$ to a distant location $B$, using only local operations, in particular homodyne detections, classical communication, and consuming as a resource $m$ two-mode squeezed vacuum states between $B$ and an auxiliary register $A'$, defined for $r\geq 0$ by~\cite[Eq.~(3.7.52)]{BARNETT-RADMORE}
\bb
\ket{\psi(r)}_{A'B} \coloneqq \frac{1}{\cosh(r)} \sum_{k=0}^\infty (-1)^k \tanh^k(r) \ket{kk}_{AB}\, ,
\label{tmsv}
\ee
where $\ket{k}$ is a local Fock state~\eqref{Fock}. The protocol consists of three steps: first, the systems $A$ and $A'$ are mixed by means of $m$ 50:50 beam splitters (each for each pair of modes); then, the quadratures $x_j$ and $p_j$ are measured on each of the $m$ pairs of modes (one in $A$ and one in $A'$, for each $j$); in general, this measurement will have a detection efficiency $\eta\in (0,1]$, with $\eta=1$ corresponding to the ideal case; finally, the outcomes are transmitted to the place where the $B$ system is, and a displacement unitary is applied on $B$ conditioned on those outcomes.

The protocol does not amount to a perfect transmission of the system $A$ into $B$, because~\eqref{tmsv} is only an approximate maximally entangled state for finite $r$. To describe in more precise terms what type of noise the transmission is subjected to, we need to introduce the \emph{Gaussian noise channel} of parameter $\lambda>0$. It is defined for a single-mode system by the equation
\bb
\NN_\lambda (\cdot) \coloneqq&\ \frac{1}{\pi\lambda} \int d^{2}\alpha\, e^{-\frac{|\alpha|^2}{\lambda}}\, \DD_\alpha (\cdot) \\
=&\ \frac{1}{\pi} \int d^{2}\alpha\, e^{-|\alpha|^2}\, \DD_{\sqrt\lambda\,\alpha} (\cdot)\, ,
\label{noise}
\ee
where $\DD_\alpha(\cdot) \coloneqq D(\alpha)(\cdot) D(\alpha)^\dag$ is the displacement unitary channel. We can now describe the whole Braunstein--Kimble protocol in terms of the map~\cite[Eq.~(8)]{Braunstein1998}
\begin{align}
\rho_{RA} \otimes \left(\psi(r)^{\otimes m}\right)_{A'B} &\longmapsto \left( I^R \otimes \big(\NN_{\lambda(r,\eta)}^{\otimes m}\big)^{A\to B}\right) \left( \rho_{RA} \right) ,
\label{BK_ancilla} \\
\lambda(r,\eta) &\coloneqq e^{-2r}+\frac{1-\eta^2}{\eta^2}\, . \label{lambda_r_eta}
\end{align}
where $R$ is an arbitrary reference system, which should be thought as an external quantum memory correlated and possibly entangled with $A$, and $\NN_{\lambda(r,\eta)}^{\otimes m}$ denotes $m$ copies of the channel~\eqref{noise} run in parallel, connecting one by one all modes of $A$ and $B$.

As we see from~\eqref{BK_ancilla}, for finite values of the squeezing parameter $r\in \R$ or non-ideal detection efficiency $\eta\in (0,1)$, the output state is not teleported perfectly, but is subjected to a noisy channel. However, for every \emph{fixed} $\rho_{RA}$ the state on the right-hand side of~\eqref{BK_ancilla} converges to $\rho_{RB}$, i.e.\ the protocol approximates a perfect teleportation, in the limit of $r\to\infty$ and $\eta\to 1^-$. It was argued in~\cite[Section~II.B]{Mark-strong-uniform} that such a convergence is strong but not uniform, that is, the values of $r$ and $\eta$ required to achieve a prescribed accuracy will depend on the input state. We provide an independent proof of this fact in Appendix~\ref{strong_not_uniform_app}. At any rate, since this dependence is not well understood and moreover a precise description of the state is often not experimentally available, it is important to have an estimate of the accuracy that is based only on few physically relevant parameters.

Our first result achieves precisely this: it provides a quantitative estimate of the maximum error introduced by the Braunstein--Kimble teleportation protocol, as expressed solely in terms of the squeezing of the resource state, the efficiency of the detectors employed, and the mean energy of the state to be teleported. 

\begin{thm} \label{BK_accuracy_thm}
Let $A,A',B$ be $m$-mode systems, and let $R$ be an arbitrary quantum system. Fix an energy threshold $E>0$, and consider a state $\rho_{RA}$ such that $\Tr\rho_A N_A - \|z\|^2 \leq E$, where $N_A = \sum_j a_j^\dag a_j$ is the total photon number operator on $A$, and $\|z\|^2\coloneqq \sum_j \left| \Tr \rho_A a_j \right|^2$. Then, the error introduced by Braunstein--Kimble teleportation of $\rho_{RA}$ over a $2m$-mode squeezed vacuum state $\left(\psi(r)^{\otimes m}\right)_{A'B}$ with detection efficiency $\eta\in (0,1]$ can be upper bounded by
\begin{align}
&\left\| \widetilde{\rho}_{RB} - \rho_{RB} \right\|_1 \nonumber \\
&\quad \leq \int_0^{+\infty} \hspace{-4ex} dx\, f(x,2m)\, \sin\!\left(\! \min\left\{\! \gamma_E \sqrt{\frac{\lambda(r,\eta)\, x}{2}},\, \frac{\pi}{2}\right\}\!\right) \label{BK_accuracy_refined} \\
&\quad \leq \frac{2\,\Gamma\left( m+1/2\right)}{(m-1)!}\, \gamma_E\, \sqrt{\lambda(r,\eta)}\, , \label{BK_accuracy}
\end{align}
where
\bb
f(x,2m) \coloneqq \frac{x^{m-1} e^{-x/2}}{2^m \Gamma(m)}
\ee
is the chi-square probability distribution in $2m$ variables, and
\begin{align}
\gamma_E &\coloneqq \sqrt{E}+\sqrt{E+1}\, , \label{gamma_E} \\
\widetilde{\rho}_{RB} &\coloneqq \left( I^R \otimes \NN_{\lambda(r,\eta)}^{A\to B} \right) (\rho_{RA})\, ,
\end{align}
with $\lambda(r,\eta)=e^{-2r}+\eta^{-2}-1$ given by~\eqref{lambda_r_eta}.
\end{thm}

\tcr{The importance of Theorem~\ref{BK_accuracy_thm} lies in the fact that it can be used to determine the values of $r$ and $\eta$ needed to reach a certain accuracy in the Braunstein--Kimble teleportation~\eqref{BK_ancilla}. Ours is the first such bound that we are aware of. Prior to our work, it was simply not clear how to determine $r$ and $\eta$ given the desired degree of precision of the overall protocol.

Note that teleportation is a fundamental primitive in a wide variety of quantum algorithms and circuits, as it allows to move around quantum information without physically displacing the systems that carry it. Unsurprisingly, a huge amount of experiments have tried to reproduce it with ever increasing fidelity~\cite{teleportation-review}. Thus, the above Theorem~\ref{BK_accuracy_thm} is likely to prove instrumental to design CV quantum circuits with prescribed error tolerance. To facilitate applications, we recall that the factor $e^{2r}$ appearing in~\eqref{BK_accuracy_refined}--\eqref{BK_accuracy} can be expressed as $e^{2r}=10^{\,s/10}$, where $s$ is the squeezing intensity measured in $\si{\decibel}$.}

An important special case of~\eqref{BK_accuracy} can be obtained by looking at the single-mode case:
\bb
\left\| \widetilde{\rho}_{RB} - \rho_{RB} \right\|_1 \leq \sqrt{\pi}\, \gamma_E\, \sqrt{\lambda(r,\eta)}\, .
\label{BK_accuracy_single_mode}
\ee
For large $m$, instead, the $m$-dependent pre-factor in~\eqref{BK_accuracy}, which scales as $\sqrt{m}$, implies that $\lambda(r,\eta) \sim 1/m$ has to be achieved to guarantee a constant accuracy. In this specific case, the problem of estimating the error introduced by CV teleportation has been tackled by Sharma et al.~\cite{Sharma2020} with different techniques. Our Theorem~\ref{BK_accuracy_thm} has the advantage of being fully analytical, expressed by means of the operationally meaningful trace norm, and valid for arbitrarily many modes.

Eq.~\eqref{BK_accuracy_refined}--\eqref{BK_accuracy} can be rephrased in terms of a notion called the `energy-constrained diamond norm'. Let $\mathcal{L} = \mathcal{L}_{A\to B} : \T (\HH_A)\to \T (\HH_B)$ be a super-operator that preserves self-adjointness. Consider a grounded Hamiltonian $H_A$ on $\HH_A$, i.e.\ a densely defined self-adjoint operator on $\HH_A$ with ground state energy $0$ (this latter requirement is purely conventional, as the energy can always be re-defined by adding a constant to it). For some energy threshold $E\geq 0$, we define the \emph{energy-constrained diamond norm} of $\mathcal{L}$ to be~\cite{PLOB, Pirandola2017, Shirokov2018, VV-diamond}
\bb
\left\| \mathcal{L} \right\|_{\diamond}^{H_{\!A},E} \coloneqq&\ \sup_{\rho_{AR}:\, \Tr [\rho_{\!A} H_{\!A}] \leq E} \left\| \left( \mathcal{L}_{A\to B} \otimes I_{\!R}\right)(\rho_{AR}) \right\|_1 \label{EC_diamond}
\ee
where the supremum is over all states $\rho_{AR}$ of an arbitrary bipartite system $AR$.

With this language, Eq.~\eqref{BK_accuracy_refined}--\eqref{BK_accuracy} imply that
\bb
&\left\| \NN_\lambda^{\otimes m} - I \right\|_\diamond^{N,E} \\
&\qquad \leq \int_0^{+\infty} \hspace{-4ex} dx\, f(x,2m)\, \sin\!\left(\! \min\left\{\! \gamma_E \sqrt{\frac{\lambda\, x}{2}},\, \frac{\pi}{2}\right\}\!\right) \\
&\qquad \leq \frac{2\,\Gamma\left( m+1/2\right)}{(m-1)!}\, \gamma_E\, \sqrt{\lambda}\, ,
\label{BK_accuracy_diamond}
\ee
with the same notation as in Theorem~\ref{BK_accuracy_thm}. (In fact,~\eqref{BK_accuracy_diamond} is essentially equivalent to~\eqref{BK_accuracy_refined}--\eqref{BK_accuracy}, as we will see at the end of this section.)

We now set out to prove Theorem~\ref{BK_accuracy_thm}. Our main mathematical tool is a recently derived estimate for the energy-constrained diamond norm distance between displacement channels~\cite[Eq.~(3)]{EC-diamond}. It reads
\bb
\left\| \DD_\alpha - \DD_\beta\right\|_\diamond^{N,E} \leq 2 \sin\left( \min\left\{ \gamma_E \|\alpha-\beta\|,\, \frac{\pi}{2}\right\}\right) ,
\label{diamond_displacements}
\ee
where $\gamma_E$ is as in~\eqref{gamma_E}, and as usual $\DD_z(\cdot) \coloneqq D(z) (\cdot) D(z)^\dag$.

\begin{proof}[Proof of Theorem~\ref{BK_accuracy_thm}]
For all $\lambda,\lambda'\in [0,1]$, let us write
\bb
&\left\| \NN_{\lambda}^{\otimes m} - \NN_{\lambda'}^{\otimes m}\right\|_\diamond^{N,E} \\
&\quad \textleq{1} \frac{1}{\pi^m} \int d^{2m}\alpha\, e^{-\|\alpha\|^2} \left\| \DD_{\sqrt\lambda\, \alpha} - \DD_{\sqrt{\lambda'}\, \alpha} \right\|_\diamond^{N,E} \\
&\quad \textleq{2} \frac{2}{\pi^m} \int d^{2m}\alpha\, e^{-\|\alpha\|^2} \\
&\hspace{17.1ex} \times \sin\left( \min\left\{ \gamma_E \left| \sqrt{\lambda} - \sqrt{\lambda'} \right| \|\alpha\|,\, \frac{\pi}{2}\right\}\right) \\
&\quad \texteq{3} 2 \int_0^{+\infty} \hspace{-4ex} dx\, f(x,2m)\\
&\hspace{11.2ex} \times \sin\left( \min\left\{ \gamma_E \left|\sqrt{\lambda} - \sqrt{\lambda'} \right| \sqrt{\frac{x}{2}},\, \frac{\pi}{2}\right\}\right)
\ee
Here, in~1 we exploited the second representation in~\eqref{noise} and applied the triangle inequality, 2~follows from the estimate~\eqref{diamond_displacements}, and finally in~3 we performed the change of variables $x\coloneqq 2\|\alpha\|^2$.

Now, a simple linear estimate can be obtained by observing that $\sin\left( \min\left\{ z, \frac{\pi}{2}\right\}\right)\leq z$ for all $z\geq 0$. Then, we see that
\bb
&\left\| \NN_{\lambda}^{\otimes m} - \NN_{\lambda'}^{\otimes m} \right\|_\diamond^{N,E} \\
&\qquad \leq \sqrt2\, \gamma_E \left| \sqrt{\lambda} - \sqrt{\lambda'}\right| \int_0^{+\infty} dx\, f(x,2m)\, \sqrt{x} \\
&\qquad \texteq{4} \frac{2\,\Gamma\left( m+1/2\right)}{(m-1)!}\, \gamma_E \left| \sqrt{\lambda} - \sqrt{\lambda'}\right| ,
\ee
where~4 follows from the definition of Gamma function.
In particular, taking $\lambda'=0$ yields~\eqref{BK_accuracy_diamond}.

We now return to the claim of Theorem~\ref{BK_accuracy_thm}. Setting $z_j\coloneqq \Tr \rho_A a_j$, we can define
\bb
\rho'_{RA} \coloneqq \Big(\id \otimes D(z)^\dag\Big)\, \rho_{RA}\,\Big(\id\otimes D(z)\Big)\, .
\ee
Thanks to~\eqref{displacement_annihilation}, one verifies immediately that
\bb
\Tr \rho'_A N_A &= \sum_j \Tr \rho_A D(z) a_j^\dag a_j D(z)^\dag \\
&=  \sum_j \Tr \rho_A \left( D(z) a_j D(z)^\dag\right)^\dag \left( D(z) a_j D(z)^\dag\right) \\
&=  \sum_j \Tr \rho_A \left(a_j - z_j\right)^\dag \left(a_j - z_j \right) \\
&= \Tr \rho_A N_A - \|z\|^2 \\
&\leq E\, .
\ee
On the other hand, using~\eqref{CCR_Weyl} it is not difficult to see that $\NN_\lambda$ commutes with the action of displacement operators, entailing that
\bb
&\left(I^R \otimes \big(\NN_{\lambda}^{\otimes m}\big)^{A\to B}\right) (\rho_{RA}) \\
&\ = \left( \id\otimes D(z) \right) \left(I^R\otimes \big(\NN_\lambda^{\otimes m}\big)^{A\to B}\right) (\rho'_{RA}) \left(\id \otimes D(z)^\dag \right) .
\ee
Thanks to the unitary invariance of the trace norm, we therefore see that
\bb
&\left\| \left(I^R\otimes \big(\NN_\lambda^{\otimes m}\big)^{A\to B}\right) (\rho_{RA}) - \rho_{RB}\right\|_1 \\
&\qquad = \left\| \left(I^R \otimes \big(\NN_\lambda^{\otimes m}\big)^{A\to B}\right) (\rho'_{RA}) - \rho'_{RB}\right\|_1 \\
&\qquad = \left\| \left(I^R \otimes \left(\big(\NN_\lambda^{\otimes m}\big)^{A\to B} - I^{A\to B}\right)\right) (\rho'_{RA})\right\|_1 \\
&\qquad \leq \left\|\big(\NN_\lambda^{\otimes m}\big)^{A\to B} - I^{A\to B}\right\|_\diamond^{N, E} \\
&\qquad = \left\|\NN_{\lambda}^{\otimes m} - I\right\|_\diamond^{N, E} ,
\ee
where naturally $\rho'_{RB}$ is just the state $\rho'$ written in the registers $RB$, and we have also made use of the definition of energy-constrained diamond norm~\eqref{EC_diamond}. The estimates in~\eqref{BK_accuracy_refined}--\eqref{BK_accuracy} then follow immediately from the already proven~\eqref{BK_accuracy_diamond} by making the substitution $\lambda=e^{-2r} + \eta^{-2} -1 = \lambda(r,\eta)$.
\end{proof}

\section{Quantum data hiding: generalities}

Quantum state discrimination is a fundamental task in quantum information processing, and a primitive of paramount operational importance~\cite{Bae2015}. In a \emph{binary quantum state discrimination} problem one of two known states $\rho$ and $\sigma$ is prepared randomly (with known a priori probabilities $p$ and $1-p$, respectively), and the task consists in guessing which one. We will refer to the triple $(\rho,\sigma; p)$ as a \emph{scheme}. A set of available strategies is modelled by a family $\MM$ of quantum measurements, mathematically described --- as detailed in the Introduction --- by positive operator-valued measures (POVM) $E(dx)$ over a measurable space $\pazocal{X}$ that in addition obey the normalisation condition $\int_\pazocal{X} E(dx)=\id$~\cite[Definition~11.29]{HOLEVO-CHANNELS-2}. Remember that the outcome of $E$ on $\rho$ is described by a random variable $X$ over $\pazocal{X}$ with probability measure $\mu_X(dx) = \Tr[\rho E(dx)]$. 

In order for the problem not to trivialise, $\MM$ has to be sufficiently rich to retain a non-zero amount of discrimination power for \emph{all} schemes $(\rho,\sigma;p)$. We capture this intuition by defining the notion of information completeness: a set of measurements $\MM$ is called \emph{informationally complete} if, for any two given states $\rho,\sigma$, the identity $\Tr[\rho E(dx)] \equiv \Tr[\sigma E(dx)]$ between measures on $\pazocal{X}$ holds for all $E\in \MM$ if and only if $\rho=\sigma$. In other words, given two states $\rho\neq \sigma$ and an informationally complete set of measurements $\MM$, we should always be able to find $E\in \MM$ that yields two different random variables when applied to $\rho$ and $\sigma$.

Now, one can show that given a binary state discrimination problem modelled by the scheme $(\rho,\sigma; p)$ and an informationally complete set of available measurements $\MM$, the lowest possible error probability for guessing the unknown state correctly is given by~\cite{VV-dh, ultimate}
\begin{align}
P_e^{\MM}(\rho,\sigma; p) =&\ \frac12 \left( 1-\left\|p\rho-(1-p)\sigma\right\|_{\MM}\right) , \label{mm_Helstrom}\\
\left\|Z\right\|_\MM \coloneqq &\ \sup_{E\in \MM} \int_\pazocal{X} \left| \Tr[Z E(dx)]\right| . \label{mm_norm}
\end{align}
Note the formal resemblance between the Holevo--Helstrom theorem~\eqref{Helstrom} and~\eqref{mm_Helstrom}, the only difference being that the trace norm is replaced by the \emph{distinguishability norm} $\left\|\cdot\right\|_\MM$ associated with $\MM$, which is a norm and not a semi-norm precisely because $\MM$ is informationally complete. Accordingly, we see immediately that the distinguishability norm associated with the set of all measurements is just the trace norm, in formula $\|\cdot\|_{\mathrm{ALL}} = \|\cdot\|_1$. The fact that $\left\|\cdot\right\|_\MM$ is a norm implies that it satisfies the triangle inequality
\bb
\left\|Z_1+Z_2\right\|_\MM \leq \|Z_1\|_\MM + \|Z_2\|_\MM
\label{triangle_mm_norm}
\ee
for all trace class operators $Z_1, Z_2\in \T(\HH)$.

As a small technical clarification, we note that in~\eqref{mm_norm} the measure $\left| \Tr[Z E(dx)]\right| = \widetilde{\nu}_Z(dx)$ is to be intended as the \emph{total variation} $\widetilde{\nu}_Z=|\nu_Z|$ of the signed measure $\nu_Z(dx) = \Tr[Z E(dx)]$, and its integral as the total variation norm $\|\nu_Z\|$ of $\nu_Z$~\cite[Section~6.6]{RUDIN-ANALYSIS}. 

In what follows, we will often refer to
\bb
\beta_\MM(\rho,\sigma; p) \coloneqq \left\|p\rho-(1-p)\sigma\right\|_{\MM}
\label{bias_mm}
\ee
as the \emph{bias} of the scheme $(\rho,\sigma;p)$ with respect to measurements in $\MM$. Note that $\beta_\MM(\rho,\sigma; p)\in [0,1]$ for all $\MM$ and all schemes. A bias close to $1$ implies that $\rho$ and $\sigma$ are almost perfectly distinguishable with measurements in $\MM$, while a low value of the bias implies that they are almost indistinguishable. Since $\beta_\MM = 1-2P_e^\MM $ and $\beta_1=1-2P_e$, the bias is just an alternative and more convenient parametrisation of the error probability. A bias with subscript $1$ indicates that we refer to the set of all measurements, i.e.
\bb
\beta_1(\rho,\sigma; p) \coloneqq \left\|p\rho-(1-p)\sigma\right\|_1\, .
\label{bias_all}
\ee
When that causes no ambiguity, we will often drop the dependence on the scheme, too.

The essence of the phenomenon called data hiding is that for certain sets $\MM$ there may exists states that are almost orthogonal --- i.e.\ almost perfectly distinguishable with general measurements --- yet almost indistinguishable when only measurements from $\MM$ are allowed. \tcr{Note that since $\MM$ is informationally complete and hence $\|\cdot\|_\MM$ is a norm, there cannot exist distinct states that are completely indistinguishable under all measurements in $\MM$; hence, all we can hope to find are \emph{sequences} of schemes that approximate this behaviour to an ever increasing degree of precision.} We give a precise definition below.

\begin{Def} \label{dh_def}
Let $\MM$ be an informationally complete set of measurements. A sequence of schemes $(\rho_n, \sigma_n; p_n)_{n\in \N}$ is said to exhibit \emph{data hiding} against $\MM$ if
\bb
\lim_{n\to\infty} \beta_1(\rho_n, \sigma_n; p_n) = 1\, ,\qquad \lim_{n\to\infty} \beta_\MM(\rho_n, \sigma_n; p_n) = 0\, , \label{dh}
\ee
where $\beta_\MM$ and $\beta_1$ are defined by~\eqref{bias_mm} and~\eqref{bias_all}, respectively. Accordingly, we say that $\MM$ itself exhibits data hiding if there exists a sequence of schemes with the property~\eqref{dh}.
\end{Def}

In light of~\eqref{Helstrom} and~\eqref{mm_Helstrom}, this amounts to saying that $\lim_n P_e(\rho_n,\sigma_n; p_n)=0$ but $\lim_n P_e^\MM (\rho_n,\sigma_n; p_n)=1/2$, capturing the intuition that the two states $\rho_n$ and $\sigma_n$ should become close to perfectly distinguishable under general measurements but almost indistinguishable under measurements in $\MM$. Thanks to~\eqref{dh}, it is not difficult to prove that the schemes $(\rho_n, \sigma_n; p_n)$ exhibit data hiding if and only if $(\rho_n, \sigma_n; 1/2)$ do as well. 
The proof of this fact is relegated to Appendix~\ref{inequivalence_norms_app} (see Lemma~\ref{equiprobable_lemma} there). That appendix expounds also an appealing geometrical interpretation of the notion of data hiding in terms of inequivalence of norms on infinite-dimensional spaces that is however not strictly necessary to the understanding of the paper.
In light of the above discussion, we will often restrict ourselves to equiprobable schemes when assessing the existence of data hiding against a certain set of measurements.

\section{Data hiding against local operations and classical communication} \label{locc_dh_sec}

\tcr{Historically, the first example of quantum data hiding that has been discovered works against the set of measurements that can be implemented by local operations and classical communication (LOCC) on a bipartite quantum system~\cite{dh-original-1, dh-original-2}. The set of LOCC protocols, of which LOCC measurements constitute a special subclass, is of paramount operational importance; for example, it defines the most common framework of entanglement manipulation~\cite{Bennett-distillation, Bennett-distillation-mixed, Bennett-error-correction, LOCC, Horodecki-review}.

Let us now illustrate the original example by Terhal, DiVincenzo, and Leung of a sequence of schemes achieving data hiding against LOCC~\cite{dh-original-1, dh-original-2}. We will simplify the presentation somewhat, as well as tweak the a priori probability so as to achieve the optimal bias gap within that particular class~\cite[Eq.~(44)--(45)]{ultimate}. On a bipartite quantum system with Hilbert space $\C^n\otimes \C^n$, consider the two extremal Werner states~\cite{Werner}, given by $\pi^\pm_n\coloneqq \frac{\id \pm F}{n(n\pm 1)}$, where $F\ket{\alpha\beta}\coloneqq \ket{\beta\alpha}$ is the swap operator, with a priori probabilities $p_n^+\coloneqq \frac{n+1}{2n}$ and $p_n^-\coloneqq 1-p_n^+ = \frac{n-1}{2n}$. While --- due to orthogonality --- $\pi^+_n$ and $\pi^-_n$ are perfectly distinguishable with global measurements, and hence $\beta_1\left(\pi_n^+, \pi_n^-; p_n^+\right) = 1$, it can be shown that
\bb
\beta_{\locc}\left(\pi_n^+, \pi_n^-; p_n^+\right) = \left\|p_n^+ \pi_n^+\! - p_n^- \pi_n^-\right\|_\locc = \frac{1}{n}\, .
\label{Werner_dh}
\ee
Since the above expression converges to $0$ as $n\to\infty$, the schemes $\left(\pi_n^+,\pi_n^-; p_n^+\right)$ exhibit data hiding against LOCC. They are also close to optimal, as it is known that~\cite{ultimate, VV-dh}
\bb
\beta_{\locc} \geq \frac{\beta_1}{2d-1} 
\label{teleportation_locc}
\ee
holds whenever the maximum local dimension is $d$. The scheme with Werner states achieves $1/n$ on the left-hand side and $1/(2n-1)$ on the right-hand side of~\eqref{teleportation_locc}, so one says that it is optimal up to a multiplicative constant (in this case, $2$).

Strictly speaking, the above Werner schemes pertain to different quantum systems. However, it is not difficult to embed them into a single bipartite Hilbert space of the form $\HH\otimes \HH$, where $\HH$ is infinite dimensional. This is also necessary if one wants to obtain data hiding schemes housed in a single system and satisfying Definition~\ref{dh_def}, precisely because of~\eqref{teleportation_locc}. Fortunately, the continuous variable (CV) systems described in Section~\ref{CV_systems_subsec}, which are very relevant for applications to quantum technologies, are already infinite dimensional. This motivates us to look for data hiding against LOCC in the context of CV systems.

By the above discussion, it should be clear that data hiding against LOCC exists in any locally infinite-dimensional system --- it suffices to embed the schemes with Werner states into it. The problem we tackle now is therefore the converse one: given some constraints on the available states, how good a data hiding can we hope to achieve? Since we are primarily interested in CV systems, the most physically motivated such constraint concerns the energy. Namely, on a bipartite quantum system $AB$, where $A$ is made of $m$ modes, how large can the gap between $\beta_1\left(\rho_{AB},\sigma_{AB}; p\right)$ and $\beta_\locc\left(\rho_{AB},\sigma_{AB}; p\right)$ be if we require that $A$ has no more than $E$ photons on average, where $E\geq 0$ is a prescribed energy budget? Mathematically, this corresponds to requiring that $\Tr \rho_A N_A,\, \Tr \sigma_A N_A\leq E$, where $N_A$ is the total photon number Hamiltonian~\eqref{total_photon_number} on $A$. Besides improving our general understanding of the phenomenon of data hiding in quantum systems, answering this question can help us to design the benchmarking protocols sketched out in the Introduction while keeping the energy consumption at a minimum.}

In order to introduce the main result of this section, we need to clarify some notation. In what follows, for a bipartite system $AB$ we will denote with $\beta_{\locc_\to}$ the bias corresponding to the set $\locc_\to$ of measurements that can be implemented by means of local operations and one-way classical communication from $A$ to $B$.


\begin{thm} \label{telep_multimode_thm}
Let $A$ be an $m$-mode CV system, and let $B$ be a generic quantum system. For a scheme $(\rho_{AB}, \sigma_{AB}; p)$ over the bipartite system $AB$ with the property that
\bb
\inf_{\alpha\in \C^m}\!\! \max\!\left\{\! \Tr\!\left[\!D(\!\alpha)\rho_{\!A} D(\!\alpha)^{\!\dag} \! N_{\!A}\!\right]\!, \Tr\!\left[\!D(\!\alpha)\sigma_{\!A} D(\!\alpha)^\dag \! N_{\!A}\!\right] \!\right\} \leq E.
\label{irreducible_photon_num}
\ee
for some $E\geq 0$ (for example, it suffices to take $E=\max\left\{ \Tr[\rho_A N_A],\, \Tr[\sigma_A N_A] \right\}$), it holds that
\begin{align}
\beta_\locc \geq \beta_{\locc_\to} \geq c_m \frac{\beta_1^{2m+1}}{\gamma_E^{2m}}\, , \label{telep_linear_multimode}
\end{align}
where
\bb
c_m \coloneqq \frac{1}{4^{m+1} m} \left( \frac{2m}{2m+1}\right)^{2m+1} \left( \frac{(m-1)!}{\Gamma(m+1/2)} \right)^{2m} ,
\label{c_m}
\ee
and as in~\eqref{gamma_E} we set $\gamma_E = \sqrt{E}+\sqrt{E+1}$. For a fixed $m$, the scaling of~\eqref{telep_linear_multimode} with $E$ is tight.
\end{thm}

\begin{rem}
Note that functions of the form $\Tr\left[D(\alpha)\rho D(\alpha)^\dag N\right]$, when finite, are second-degree polynomials in $\alpha\in \C^m$. In fact,
\bb
&\Tr\left[D(\alpha)\rho D(\alpha)^\dag N\right] \\
&\qquad = \Tr\left[\rho\, D(\alpha)^\dag \left( \sumno_j a_j^\dag a_j\right) D(\alpha)\right] \\
&\qquad = \Tr\left[\rho \left( \sumno_j (a_j + \alpha_j)^\dag (a_j+\alpha_j) \right) \right] \\
&\qquad = \Tr[\rho N] + \sum_j \left(\alpha_j^* \Tr[\rho a_j] + \alpha_j \Tr[\rho a_j^\dag]\right) + \left\|\alpha\right\|^2  ,
\ee
where the second equality follows from~\eqref{displacement_annihilation}. Therefore, for any given pair of reduced states $\rho_A$ and $\sigma_A$, computing the infimum in~\eqref{irreducible_photon_num} (which is in fact a minimum) is in principle an elementary task.
\end{rem}


Theorem~\ref{telep_multimode_thm} identifies a precise trade-off between data hiding effectiveness and energy investment, for any number of modes. As we saw, such scaling is provably tight. However, the coefficient $c_m$ in~\eqref{c_m} decreases very rapidly in $m$ (for example, we have $c_m<10^{-6}$ for $m\geq 4$). We leave open the problem of finding a bound with a better scaling in $m$. At any rate, this discussion should convince the reader that Theorem~\ref{telep_multimode_thm} is most useful when $m$ is small. Hence, we rephrase its statement in the important special case where $m=1$ as a corollary, which anyway contains all the fundamental physics of the problem.

\begin{cor} \label{telep_thm}
Let $A$ be a single-mode CV system with photon number operator $N_A$, and let $B$ be a generic quantum system. For a scheme $(\rho_{AB}, \sigma_{AB}; p)$ over the bipartite system $AB$ with the property that $\Tr[\rho_A N_A],\, \Tr[\sigma_A N_A]\leq E$, it holds that
\begin{align}
\beta_\locc \geq \beta_{\locc_\to} \geq \frac{2}{27\pi}\, \frac{\beta_1^3}{\gamma_E^2} \, , \label{telep_linear}
\end{align}
where $\gamma_E$ is given by~\eqref{gamma_E}. The scaling of~\eqref{telep_linear} with $E$ is tight.
\end{cor}

Note that~\eqref{telep_linear} can be rephrased as 
\bb
1-2P_e^\locc \geq \frac{2 \left( 1-2P_e\right)^{3}}{27\pi \left( \sqrt{E}+\sqrt{E+1}\right)^2}
\label{telep_explicit_errors}
\ee
in terms of the error probabilities in state discrimination.

Let us give a quick example of how to use the above Corollary~\ref{telep_thm} in applications. To fix ideas, let us consider only schemes for which $\beta_1=1$, i.e.\ the hidden bit can be retrieved perfectly by means of global measurements. Then,~\eqref{telep_linear} implies e.g.\ that in order to achieve $\beta_\locc < 10^{-3}$ we need to prepare states with mean energy $E > 5.4$.

\tcr{We now set out to prove Theorem~\ref{telep_multimode_thm} and Corollary~\ref{telep_thm}. This is where Theorem~\ref{BK_accuracy_thm} above will play a decisive role.
In fact, the crux of our proof of Theorem~\ref{telep_multimode_thm} is the following CV generalisation of the `teleportation argument', a technique to lower bound the $\locc_\to$ distinguishability norm first presented in~\cite[Theorem~16]{ultimate}, where it is used to investigate the finite-dimensional case.
As the finite-dimensional version of the argument relied crucially upon the existence of quantum teleportation, our general statement exploits the CV generalisation of this procedure, namely, the Braunstein--Kimble teleportation protocol, and the quantitative analysis of it we conducted in Theorem~\ref{BK_accuracy_thm}. In what follows, for a bipartite system $AB$, where $A$ is held by Alice and $B$ by Bob, we denote with $\|\cdot\|_{\locc_\to}$ the distinguishability norm associated with the set of measurements $\locc_\to$ that can be implemented by local quantum operations assisted by one-way classical communication from Alice to Bob.}

\begin{lemma}[(Teleportation argument -- continuous variable case)] \label{teleportation_argument_CV_lemma}
Let $A=A_1\ldots A_m$ be an $m$-mode CV system, and let $B$ be an arbitrary quantum system. For any trace class operator $Z_{AB}$ on $AB$, it holds that
\bb
&\left\|Z_{AB}\right\|_{\locc} \\[1ex]
&\quad \geq \left\|Z_{AB}\right\|_{\locc_\to} \\
&\quad \geq \sup_{0<\lambda_1,\ldots, \lambda_m \leq 1} \frac{\prod_j \!\lambda_j}{2\!-\!\prod_j\!\lambda_j}\left\|\! \left( \!\bigotimes\nolimits_j\! \NN_{\!\lambda_j}^{A_j} \!\otimes I^B\right)\! (Z_{AB})\right\|_1 ,
\label{teleportation_argument_CV}
\ee
where the direction of communication goes from $A$ to $B$.
\end{lemma}

Before proving Lemma~\ref{teleportation_argument_CV_lemma}, we need to fix some notation. The standard entanglement robustness of a generic bipartite state $\rho_{AB}$ is defined by~\cite{VidalTarrach, taming-PRL, taming-PRA}
\bb
R_{\pazocal{S}}^s(\rho_{\!A\!B}) \coloneqq \inf\Big\{ &R:\, \exists\, \sigma_{\!A\!B}\!\in\! \pazocal{S}_{A:B}\!:\\
&\frac1R\,\rho_{\!A\!B} + \left(1-\frac1R\right) \sigma_{\!A\!B}\in \pazocal{S}_{A:B} \Big\}\, ,
\ee
where $\pazocal{S}_{A:B}$ denotes the set of separable states on $AB$~\cite{Werner, Horodecki-review}. For a pure state $\ket{\Psi}_{AB}$ with Schmidt decomposition $\ket{\Psi}_{AB}=\sum_k \sqrt{\lambda_k} \ket{e_k}_A \ket{f_k}_B$ this quantity can be computed exactly: it evaluates to~\cite{VidalTarrach, taming-PRL, taming-PRA}
\bb
R_{\pazocal{S}}^s(\Psi_{AB}) = \left( \sumno_k \sqrt{\lambda_k} \right)^2 .
\label{rob_pure}
\ee
The above formula was proved in~\cite{VidalTarrach} for finite-dimensional systems, and recently rigorously extended to the case of infinite-dimensional ones~\cite{taming-PRL, taming-PRA}.

\begin{rem}
We warn the reader that some of the literature, including the original paper by Vidal and Tarrach, defines the standard entanglement robustness as $R_{\pazocal{S}}-1$. Our notation reflects instead more recent conventions~\cite{taming-PRL, taming-PRA}.
\end{rem}

\begin{proof}[Proof of Lemma~\ref{teleportation_argument_CV_lemma}]
From the definition~\eqref{mm_norm} and from the inclusion $\locc_\to\subseteq \locc$ we see immediately that
\bb
\left\|Z_{AB}\right\|_{\locc} \geq \left\|Z_{AB}\right\|_{\locc_\to}
\ee
holds. Now, let $r_1,\ldots, r_m \geq 0$ be parameters. Thanks to~\eqref{rob_pure}, the standard robustness of entanglement of the product
\bb
\ket{\Psi}_{\!A'\!B'}\! \coloneqq&\, \bigotimes_j \ket{\psi(r_j)}_{A'_jB'_j} \\
=&\, \frac{1}{\prod_j \cosh(r_j)}\! \sum_{k\in \N^m} \!\prod_j \left( - \!\tanh r_j \right)^{\sum_j \!k_j} \!\ket{k}_{\!A'} \!\ket{k}_{\!B'}
\ee
of two-mode squeezed vacuum states~\eqref{tmsv} evaluates to
\bb
R \coloneqq&\ R_{\pazocal{S}}^s\left( \Psi_{A'B'} \right) \\
=&\ \left( \sum_{k\in \N^m} \prod_j \frac{(\tanh r_j)^{\sum_j k_j}}{\cosh r_j} \right)^2 \\
=&\ \prod_j \left(\sum_{k_j=0}^\infty \frac{\tanh^{k_j} r_j}{\cosh r_j} \right)^2 = \prod_j e^{2r_j} \, .
\ee
Thanks to the results of~\cite{taming-PRL, taming-PRA}, we know that there exists a separable state $\sigma_{A'B'}$ with the property that $\frac1R\, \Psi_{A'B'} + \frac{R-1}{R}\, \sigma_{A'B'}$ is also separable. The idea now is that we let Alice and Bob construct this latter state; this can indeed be done with LOCC because of separability. Now, with a certain (small) probability $1/R$ Alice and Bob hold the state $\Psi_{A'B}$. By using it as a resource in the Braunstein--Kimble protocol, they will achieve a high-fidelity teleportation with said small probability. In the final stage, Bob, who now holds both subsystems, performs the optimal state discrimination locally. The above intuitive description can be formalised mathematically as follows:
\bb
&\left\|Z_{AB}\right\|_{\locc_\to} \\
&\qquad \texteq{1} \left\| Z_{AB}\otimes \left( \frac1R\, \Psi_{A'B'} + \frac{R\!-\!1}{R}\, \sigma_{A'B'} \right) \right\|_{\locc_\to} \\
&\qquad \textgeq{2} \frac1R \left\| Z_{AB}\otimes \Psi_{A'B'} \right\|_{\locc_\to} \\
&\hspace{7.5ex} - \frac{R\!-\!1}{R}\left\| Z_{AB} \otimes \sigma_{A'B'} \right\|_{\locc_\to} \\
&\qquad \texteq{3} \frac1R \left\| Z_{AB}\otimes \Psi_{A'B'} \right\|_{\locc_\to}\! - \frac{R\!-\!1}{R} \left\| Z_{AB} \right\|_{\locc_\to} \\
&\qquad \textgeq{4} \frac1R \left\| \left( \bigotimes\nolimits_j \NN_{e^{-2r_j}}^{A_j\to B'_j}\otimes I^B\right) (Z_{AB}) \right\|_{\locc_\to} \\
&\hspace{7.5ex} - \frac{R\!-\!1}{R} \left\| Z_{AB} \right\|_{\locc_\to} \, .
\ee
Here, in~1 Alice and Bob use $\locc_\to$ operations to create the separable state on the right-hand side: doing so cannot result in any increase of the $\locc_\to$ norm, simply because an $\locc_\to$ protocol followed by an $\locc_\to$ measurement can be thought of as an effective $\locc_\to$ measurement; at the same time, it causes no loss of information, because the systems $A',B'$ can always be discarded. In~2 we simply applied the triangle inequality~\eqref{triangle_mm_norm} to the case of the $\locc_\to$ norm. In~3 we observed once more that $\left\| Z_{AB} \otimes \sigma_{A'B'} \right\|_{\locc_\to} = \left\| Z_{AB} \right\|_{\locc_\to}$, because $\sigma_{A'B'}$ is separable. Finally, in~4 we applied a particular $\locc_\to$ procedure, namely, the Braunstein--Kimble protocol~\eqref{BK_ancilla} with $\eta=1$, to the first term.
Massaging the obtained inequality we arrive at
\bb
&\left\|Z_{AB}\right\|_{\locc_\to} \\
&\quad \geq \frac{1}{2R-1} \left\| \left( \bigotimes\nolimits_j \NN_{e^{-2r_j}}^{A_j\to B'_j}\otimes I^B\right) (Z_{AB}) \right\|_{\locc_\to} \\
&\quad = \frac{\prod_j e^{-2r_j}}{2\!-\!\prod_j e^{-2r_j}} \left\| \left( \!\bigotimes\nolimits_j\! \NN_{e^{-2r_j}}^{A_{\!j}\to B'_{\!j}}\!\otimes I^B\right)\! (Z_{AB}) \right\|_{\locc_\to} .
\ee
Taking the supremum over $r_1,\ldots, r_m\geq 0$ and performing the change of variables $\lambda_j \coloneqq e^{-2r_j}$ yields~\eqref{teleportation_argument_CV}.
\end{proof}

We are now ready to present the proof of Theorem~\ref{telep_multimode_thm}.

\begin{proof}[Proof of Theorem~\ref{telep_multimode_thm}]
As usual, the first inequality $\beta_\locc\geq \beta_{\locc_\to}$ follows elementarily from the inclusion $\locc_\to\subseteq \locc$. We now move on the proof of the second. Since a displacement on $A$ is a local operation that can always be absorbed into the subsequent $\locc_\to$ measurement, it clearly holds that
\bb
&\left\| D(\alpha)_{\!A}\!\otimes\! \id_{\!B} \left( p\, \rho_{AB} - (1\!-\!p)\, \sigma_{AB}\right) D(\alpha)_{\!A}^\dag\!\otimes\! \id_{\!B}\right\|_{\locc_\to} \\
&\qquad = \left\| p\, \rho_{AB} - (1\!-\!p)\, \sigma_{AB}\right\|_{\locc_\to} .
\ee
Even more trivially, the same identity holds if the $\|\cdot\|_{\locc_\to}$ norm is replaced by the trace norm. Therefore, without loss of generality we can directly assume that
\bb
\max\left\{ \Tr[\rho_A N_A],\, \Tr[\sigma_A N_A]\right\}\leq E\, .
\ee
Setting $X_{AB}\coloneqq p\, \rho_{AB} - (1-p)\, \sigma_{AB}$, we now write that
\bb
&\beta_{\locc_\to} \\[.8ex]
&\quad = \left\| X_{AB} \right\|_{\locc_\to} \\
&\quad \textgeq{1} \sup_{0<\lambda_1,\ldots, \lambda_m \leq 1} \frac{\prod_j\! \lambda_j}{2\!-\!\prod_j\!\lambda_j} \left\| \left(\! \bigotimes\nolimits_j\! \NN_{\!\lambda_j}^{A_j}\! \otimes\! I^B\!\right) (X_{AB})\right\|_1 \\
&\quad \textgeq{2} \sup_{0<\lambda\leq 1} \frac{\lambda}{2} \left\| \left( \NN_{\lambda^{1\!/\!m}}^{\otimes m} \otimes I\right) (X_{AB})\right\|_1 \\
&\quad \textgeq{3} \sup_{0<\lambda\leq 1} \frac{\lambda}{2} \left( \left\| X_{AB} \right\|_1 -\left\| \left( \left(\NN_{\lambda^{1\!/\!m}}^{\otimes m}\! -\! I\right)\otimes I \right) (X_{AB})\right\|_1 \right) \\
&\quad \textgeq{4} \sup_{0<\lambda\leq 1} \frac{\lambda}{2} \Big( \beta_1 - p \left\| \left( \left(\NN_{\lambda^{1\!/\!m}}^{\otimes m}\! -\! I\right)\otimes I \right) (\rho_{AB})\right\|_1 \\
&\hspace{14.5ex} - (1\!-\!p) \left\| \left( \left(\NN_{\lambda^{1\!/\!m}}^{\otimes m}\! -\! I\right)\otimes I \right) (\sigma_{AB})\right\|_1\Big) \\
&\quad \textgeq{5} \sup_{0<\lambda\leq 1} \frac{\lambda}{2} \left( \beta_1 - \frac{2\,\Gamma\left( m+1/2\right)}{(m-1)!}\, \gamma_E\, \lambda^{\frac{1}{2m}}\right) \\
&\quad \texteq{6} \frac{1}{4^{m+1} m} \left( \frac{2m}{2m\!+\!1}\right)^{2m+1}\!\! \left( \frac{(m-1)!}{\Gamma(m\!+\!1/2)} \right)^{2m} \!\frac{\beta_1^{2m+1}}{\gamma_E^{2m}}\, .
\ee
Here, 1~is just an application of~\eqref{teleportation_argument_CV}, in~2 we made the ansatz $\lambda_j\equiv \lambda^{1/m}$ and simplified $\lambda/(2-\lambda)\geq \lambda/2$, 3~and~4 are simply the triangle inequality, 5~follows from~\eqref{BK_accuracy}, and in~6 we used the easily verified formula
\bb
\sup_{0<\lambda\leq 1} \frac{\lambda}{2}\left(a - b\, \lambda^{\frac{1}{2m}} \right) = \frac{1}{4m }\left( \frac{2m}{2m\!+\!1}\right)^{2m+1} \!\frac{a^{2m+1}}{b^{2m}}\, ,
\ee
valid for all $a\geq 0$ and $b>0$. This proves~\eqref{telep_linear_multimode}.

We now show that the scaling of~\eqref{telep_linear_multimode} in $E$ is tight for all fixed values of $m$. To do this, for all integers values of $E$ we show how to construct a scheme $(\rho_{AB}, \sigma_{AB}, p)$, with $A,B$ composed of $m$ modes each, such that $\beta_1=1$ but
\bb
\beta_{\locc} = \beta_{\locc_\to} \leq \frac{C_m}{E^m}\, ,
\ee
with $C_m$ depending only on $m$. Note that also the right-hand side of~\eqref{telep_linear_multimode} scales as $E^{-m}$. 

We start by recalling the following combinatorial identity: for all integers $E$, the number of $m$-tuples of non-negative integers that add up to no more than $E$ equals
\bb
D_m(E) \coloneqq&\, \left| \left\{ (p_1, \ldots, p_m)\!\in\! \N^m\!\!: \sumno_{j=1}^m p_j \leq E\right\} \right| \\
=&\, \left| \left\{ (p_0, p_1, \ldots, p_m)\!\in\! \N^{m+1}\!\!: \sumno_{j=0}^m p_j = E\right\} \right| \\
=&\, \binom{m+E}{m} \geq \frac{E^m}{m!}\, .
\ee
We can also interpret $D_m(E)$ as the number of Fock states~\eqref{Fock} in an $m$-mode system that have total photon number at most $E$. In other words, calling $\HH_{m,E}\subset \HH_A$ their linear span, we have that $\dim \HH_{m,E} = D_m(E)$, or $\HH_{m,E}\simeq \C^{D_m(E)}$. Now, we mentioned already (cf.~\eqref{Werner_dh}), on a bipartite quantum system of local dimension $D_m(E)$ there exists a scheme based on Werner states that satisfies $\beta_1=1$ and~\cite[Eq.~(45)]{ultimate}
\bb
\beta_{\locc} = \beta_{\locc_\to} \leq \frac{1}{D_m(E)} \leq \frac{m!}{E^m}\, .
\ee
Therefore, the scaling of the right-hand side of~\eqref{telep_linear_multimode} with respect to $E$ is tight. This completes the proof.
\end{proof}

It is easy to verify that the $m=1$ case of~\eqref{telep_linear_multimode} yields precisely~\eqref{telep_linear} (and in turn~\eqref{telep_explicit_errors}). This proves also Corollary~\ref{telep_thm}.

\begin{note}
I am grateful to an anonymous referee at the 16th Conference on the Theory of Quantum Computation, Communication and Cryptography (TQC 2021) for correctly remarking that an inequality similar to~\eqref{telep_linear_multimode}, less tight but featuring the same optimal scaling in $E$, can be derived with entirely different methods, that is, by truncating the Hilbert space (using Markov's inequality) and then exploiting directly~\eqref{teleportation_locc}.
\end{note}

\section{Data hiding against Gaussian operations and classical computation}

We now set out to investigate data hiding against GOCC measurements.
The simplest example of a state discrimination task to be solved with GOCC is the one considered by Takeoka and Sasaki, consisting in the discrimination of two equi-probable single-mode coherent states --- without loss of generality $\ket{\alpha}$ and $\ket{-\alpha}$~\cite{Takeoka-Sasaki}. In this case, one can compute both $\beta_1$ and $\beta_\gocc$ exactly, finding $\beta_1=\sqrt{1-e^{-4|\alpha|^2}}$ and $\beta_\gocc=\erf\left(\sqrt2 |\alpha|\right)$~\cite{Takeoka-Sasaki, KK-VV-GOCC}. Remarkably, this shows that GOCC can be strictly less powerful than general measurements at state discrimination. However, it does not lead to data hiding as defined by~\eqref{dh}, because there are no sequences of values of $\alpha$ that make $\beta_1\to 1$ and $\beta_\gocc\to 0$.

\tcr{In an attempt to find genuine data hiding against GOCC, we shall now consider several examples of candidate schemes of increasing complexity, described in Sections~\ref{thermal_ex}--\ref{even_odd_thermal_ex}. As we will see, only the last construction (Section~\ref{even_odd_thermal_ex}) truly achieves the sought GOCC data hiding. The reason for presenting all three of them, and in this order, is that the path we are about to take is an instructive one: if nothing else, it serves to illustrate how surprisingly difficult it is to hide some information from Gaussian measurements; indeed, most constructions relying on states that are easy to produce in practice will \emph{not} exhibit this property.

Before we can proceed we need to fix some terminology. In what follows, for a fixed scheme $(\rho,\sigma;p)$ over an $m$-mode system (typically, $m=1$), we will denote with $\beta_\hom$ (respectively, $\beta_\het$) the bias corresponding to the set composed of all homodyne measurements (respectively, of the heterodyne measurement). Since these are all GOCC measurements, we have that
\bb
\beta_\gocc \geq \max\left\{\beta_\hom, \beta_\het \right\} .
\label{GOCC_vs_hom_het}
\ee
We will make use of this inequality rather frequently in what follows, as it provides a handy way of lower bounding the GOCC bias. What if we want instead to upper bound it?}

In this case, a useful mathematical tool has been constructed by Sabapathy and Winter, who have proved that~\cite{KK-VV-GOCC, VV-GOCC-IHP, VV-GOCC-APS, VV-GOCC-ISIT}
\bb
\left\|Z\right\|_\gocc \leq \left\|Z\right\|_\Wp \leq \left\| W_Z\right\|_{L^1} \coloneqq \int d^{2m}\alpha \left|W_Z(\alpha)\right|
\label{GOCC_vs_WL1}
\ee
for all trace class operators $Z$, where $\|\cdot\|_\Wp$ denotes the distinguishability norm associated with the set of measurements having non-negative Wigner distribution~\cite{Mari-Eisert}, and $W_Z$ is the Wigner function of $Z$ (cf.~\eqref{Wigner_1}--\eqref{Wigner_3}). Essentially, the reason why the above upper bound holds is that the measurement operators corresponding to a GOCC protocol have non-negative Wigner distribution.

\subsection{Thermal states} \label{thermal_ex}

A thermal state of a single mode takes the form
\bb
\tau_\nu \coloneqq \frac{1}{\nu+1} \sum_{k=0}^\infty \left( \frac{\nu}{\nu+1}\right)^k \ketbra{k}\, ,
\label{thermal}
\ee
where $\nu\coloneqq \Tr[\tau_\nu N]$ is the mean photon number of $\tau_\nu$. Since thermal states are experimentally easy to produce, we look at schemes of the form $(\tau_\nu, \tau_\mu, 1/2)$, where without loss of generality $\mu>\nu\geq 0$. Our immediate goal is to derive exact expressions or estimates for the biases $\beta_1$, $\beta_\Wp$, $\beta_\het$, and $\beta_\hom$ of these schemes. From those we will infer upper and lower bounds for $\beta_\gocc$ as well, according to~\eqref{GOCC_vs_hom_het} and~\eqref{GOCC_vs_WL1}.

We start by reminding the reader that the Wigner and the Husimi function of a thermal state are given by~\cite[Eq.~(4.5.36)]{BARNETT-RADMORE}
\begin{align}
W_{\tau_\nu}(\alpha) &= \frac{2}{\pi(2\nu+1)}\, e^{-\frac{2|\alpha|^2}{2\nu+1}}\, ,
\label{Wigner_thermal} \\
Q_{\tau_\nu}(\alpha) &= \frac1\pi \braket{\alpha|\tau_\nu|\alpha} = \frac{1}{\pi(\nu+1)} \, e^{-\frac{|\alpha|^2}{\nu+1}}\, , \label{Husimi_thermal}
\end{align}
for all $\alpha\in\C$. We now have the following.

\begin{lemma} \label{no_dh_thermal_lemma}
For $\mu>\nu\geq 0$, the thermal states $\tau_\nu,\, \tau_\mu$ defined by~\eqref{thermal} satisfy that
\begin{align}
\frac12 \|\tau_\nu\! -\! \tau_\mu \|_1 &= \left(\!\frac{\mu}{\mu\!+\!1}\!\right)^{\!\!N_{0}(\nu,\mu)+1}\hspace{-2ex} - \left(\!\frac{\nu}{\nu\!+\!1}\!\right)^{\!\!N_{0}(\nu,\mu)+1} \!\!\! , \label{td_thermal} \\[1ex]
\frac12 \left\|\tau_\nu\! -\! \tau_\mu\right\|_{\Wp} &\leq \frac12 \left\| W_{\tau_\nu} - W_{\tau_\mu} \right\|_{L^1} \nonumber \\
&= \frac{2(\mu-\nu)}{2\mu+1} \left(\frac{2\mu+1}{2\nu+1}\right)^{-\frac{2\nu+1}{2(\mu-\nu)}}\! , \label{wigd_thermal} \\[1ex]
\frac12 \left\|\tau_\nu\! -\! \tau_\mu\right\|_\het &= \frac{\mu-\nu}{\mu+1} \left(\frac{\mu+1}{\nu+1}\right)^{-\frac{\nu+1}{\mu-\nu}} ,\label{hetd_thermal} \\[1ex]
\frac12 \left\| \tau_\nu\! -\! \tau_\mu\right\|_\hom\!\! &= \erf\!\left(\!\sqrt{\frac{2\mu\!+\!1}{4(\mu\!-\!\nu)} \ln\!\left( \frac{2\mu\!+\!1}{2\nu\!+\!1} \right)} \right) \nonumber \\
&\quad - \erf\!\left(\!\sqrt{\frac{2\nu\!+\!1}{4(\mu\!-\!\nu)} \ln\!\left( \frac{2\mu\!+\!1}{2\nu\!+\!1} \right)} \right) , \label{homd thermal}
\end{align}
where
\bb
N_0(\nu,\mu)\coloneqq \floor{\frac{\ln (\mu+1) - \ln (\nu+1)}{\ln (\mu(\nu+1)) - \ln (\nu(\mu+1))}} .
\label{N_0}
\ee
\end{lemma}

The proof of the above result is rather lengthy but ultimately mechanical; we defer it to Appendix~\ref{thermal_app}. What is more interesting for us is that now that we have Lemma~\ref{no_dh_thermal_lemma} at hand we can ask ourselves whether there is a sequence of values of $(\nu,\mu)$ that makes $\beta_1\to 1$ and $\beta_\gocc\to 0$. This turns out not to be possible. In fact, the following holds.

\begin{prop} \label{no_gocc_dh_thermal_prop}
For all $\mu,\nu\geq 0$, it holds that
\bb
\left\|\tau_\nu - \tau_\mu\right\|_{\mathrm{het}} &\leq \left\|\tau_\nu - \tau_\mu\right\|_{\gocc} \leq \left\|\tau_\nu - \tau_\mu\right\|_1 \\
&\leq e \left\|\tau_\nu - \tau_\mu\right\|_{\mathrm{het}} \leq e \left\|\tau_\nu - \tau_\mu\right\|_{\gocc} \, ,
\ee
where the Neper constant $e$ is optimal. Equivalently, any scheme of the form $(\tau_\nu, \tau_\mu, 1/2)$ satisfies that
\bb
\beta_{\mathrm{het}} \leq \beta_{\gocc} \leq \beta_1 \leq e\, \beta_{\mathrm{het}} \leq e\, \beta_\gocc\, .
\ee
In particular, such schemes cannot exhibit data hiding against GOCC.
\end{prop}

\begin{proof}
We can assume without loss of generality that $\mu > \nu\geq 0$. The only non-trivial inequality to be proved is
\bb
\left\|\tau_\nu - \tau_\mu\right\|_1 \leq e \hetd{\tau_\nu}{\tau_\mu}\, .
\label{trace_norm_vs_hetd_thermal}
\ee
Consider that
\bb
\frac12 \left\|\tau_\nu - \tau_\mu\right\|_1 &\texteq{1} \max_{N\in \N,\, N\geq 1} \left\{\! \left(\!\frac{\mu}{\mu\!+\!1}\!\right)^{\!\!N} - \left(\!\frac{\nu}{\nu\!+\!1}\!\right)^{\!\!N} \right\} \\
&\texteq{2} \max_{N\in \N,\, N\geq 1} \left\{\! \left(\!1-\frac{t}{\nu\!+\!1}\!\right)^{\!\!N} - \left(\!\frac{\nu}{\nu\!+\!1}\!\right)^{\!\!N} \right\} \\
&\textleq{3} 1-t \\
&\textleq{4} e\, (1-t) t^{\frac{t}{1-t}} \\
&\texteq{5} \frac{e}{2} \left\|\tau_\nu - \tau_\mu\right\|_{\mathrm{het}} \, .
\ee
Here, 1~is just a re-parametrisation of~\eqref{no_dh_thermal_proof_eq1}. In~2 we defined the new variable $t\coloneqq \frac{\nu+1}{\mu+1}\in [0,1]$. To derive the inequality in~3, observe that the function $f_{\nu,N} : [0,1]\to \R$ defined by
\bb
f_{\nu,N}(t)\coloneqq \left(1-\frac{t}{\nu+1}\right)^N - \left(\frac{\nu}{\nu+1}\right)^N - 1 + t
\ee
is convex and satisfies that
\bb
f_{\nu,N}(0) = - \left(\frac{\nu}{\nu+1}\right)^N \leq f_{\nu,N}(1)=0\, ;
\ee
hence, $f_{\nu,N}(t)\leq 0$ for all $t\in [0,1]$. The relation in~4 follows by a careful analysis of the function $\varphi:[0,1]\to \R$ defined by
\bbb
\varphi(t) \coloneqq (1-t) \left(e\, t^{\frac{t}{1-t}} - 1\right)
\eee
for $t<1$, with $\varphi(1) = 0$. A straightforward calculation shows that its derivatives are given by
\begin{align*}
\varphi'(t) &= 1 + \frac{e\log t}{1-t}\, t^{\frac{t}{1-t}}\, ,\\
\varphi''(t) &= \frac{e\, t^{\frac{t}{1-t}}}{t(1\!-\!t)^3} \left( (1\!-\!t \!+\! t\log t)^2 + t(1\!-\!t)(\log t)^2 \right) \geq 0\, .
\end{align*}
Note that $\varphi'(1) = 0$. Since $\varphi$ is convex, it lies above its tangent at $t=1$, in formula
\bbb
\varphi(t) \geq \varphi(1) - (1-t) \varphi'(1) = 0\, .
\eee
This justifies~4. Finally, 5~is just a rephrasing of~\eqref{hetd_thermal}. This concludes the proof.
\end{proof}

\begin{rem}
A numerical investigation seems to suggest that the non-trivial relation
\bb
\left\|W_{\tau_\nu} - W_{\tau_\mu} \right\|_{L^1} \leq \left\|\tau_\nu - \tau_\mu\right\|_1
\label{wigd_vs_trace_norm_thermal}
\ee
holds as well. We do not dwell on this further --- nor do we attempt an analytical proof of this fact --- because it appears to us a mere mathematical coincidence without special physical significance. 
\end{rem}

The above Proposition~\ref{no_gocc_dh_thermal_prop} allows us to conclude that simple schemes made of thermal states --- these being among those that can be produced more easily in a laboratory --- cannot work for data hiding against GOCC. Therefore, we move on to something a little more refined.


\subsection{Fock states} \label{Fock_ex}

Our second attempt involves discrimination of two consecutive Fock states. We consider the family of schemes $\left( \ketbra{n}, \ketbra{n\!+\!1}, 1/2\right)$. Since $\ket{n}$ and $\ket{n+1}$ are orthogonal, it holds that $\beta_1=1$ for all $n$: perfect discrimination is achievable with photon counting. On the other hand, measuring for instance the quadrature $x$ (homodyne detection) produces two probability distributions shaped as consecutive squared Hermite functions $\psi_n^2(x)$ and $\psi_{n+1}^2(x)$. Remembering that $\psi_n^2(x)$ has standard deviation $\sqrt{n+1/2}$ and a fine-grained structure with $n$ troughs and $n+1$ peaks, one could intuitively expect $\psi_n^2(x)$ and $\psi_{n+1}^2(x)$ to be almost indistinguishable for large $n$. Rather surprisingly, this is not the case: in fact, what happens is that the slight difference in the positioning of the peaks and troughs of $\psi_n^2(x)$ and $\psi_{n+1}^2(x)$ turns out to be enough to make their $L^1$-distance larger than a constant for all $n$. We formalise this in the coming statement.

\begin{prop} \label{no_dh_Fock_prop}
It holds that
\bb
\liminf_{n\to\infty} \left\| \ketbra{n} - \ketbra{n+1} \right\|_\hom \geq \frac{8}{\pi^2}\, .
\label{consecutive_Fock_L1}
\ee
In particular, there exists a universal constant $c>0$ such that 
\bb
&\left\| \ketbra{n} - \ketbra{n\!+\!1} \right\|_{\gocc} \\
&\qquad \geq \left\| \ketbra{n} - \ketbra{n\!+\!1} \right\|_\hom \geq\, c
\label{consecutive_Fock_hom}
\ee
holds for all integers $n\geq 0$. This entails that pairs of consecutive Fock states can never exhibit data hiding against GOCC.
\end{prop}

The proof of the above result is quite technical; we defer it to Appendix~\ref{Fock_app}. What this mathematical result implies is that homodyne detection yields a non-vanishing bias for all $n$ on the schemes $\left( \ketbra{n}, \ketbra{n\!+\!1}, 1/2\right)$, i.e.
\bb
\beta_\gocc \geq \beta_\hom \geq c/2 > 0\, .
\ee
In particular,
\bb
\beta_\gocc \geq \beta_{\mathrm{hom}} \geq \frac{2}{c}\, \beta_1\, ,
\ee
entailing that no schemes involving consecutive Fock states can exhibit data hiding against GOCC. We therefore move on to still more sophisticated ones.

\begin{rem}
A numerical analysis of the involved integrals seems to suggest that one could in fact take $c=8/\pi^2$ in Proposition~\ref{no_dh_Fock_prop}.
\end{rem}


\subsection{Data hiding against GOCC with even and odd thermal states} \label{even_odd_thermal_ex}

The two above examples confirm that a reasonable accuracy in state discrimination can be achieved with simple Gaussian measurements for several cases of physical interest. While we know that data hiding against GOCC does appear when the number of modes is asymptotically large~\cite{KK-VV-GOCC, VV-GOCC-IHP, VV-GOCC-APS, VV-GOCC-ISIT}, the reader may wonder whether it exists at all for single-mode systems. We now construct an explicit example showing that this is indeed the case.

For $\lambda\in [0,1)$, define the states
\begin{align}
\omega_\lambda^+ &\coloneqq (1-\lambda^2)\sum_{n=0}^\infty \lambda^{2n} \ketbra{2n}\, , \label{omega_+} \\
\omega_\lambda^- &\coloneqq (1-\lambda^2)\sum_{n=0}^\infty \lambda^{2n} \ketbra{2n+1}\, , \label{omega_-}
\end{align}
and consider the scheme $\left(\omega_\lambda^+, \omega_\lambda^-, 1/2 \right)$. This can be implemented with a simple experimental procedure starting from a two-mode squeezed vacuum state $\ket{\psi(r)}_{AB}$. By measuring the photon number \emph{parity} on system $B$~\cite{Bollinger1996, Anisimov2010, Olson2015, Birrittella2020}, corresponding to the two-outcome POVM
\bb
\left\{\sum_{\text{$n$ even}} \ketbra{n}_B,\, \sum_{\text{$n$ odd}} \ketbra{n}_B \right\} ,
\ee
we are left with the states $\omega_\lambda^+$ (even) or $\omega_\lambda^-$ (odd), where $\lambda=\tanh(r)$.

We now look at the state discrimination properties of the scheme $\left(\omega_\lambda^+, \omega_\lambda^-, 1/2 \right)$. On the one hand, since $\omega_\lambda^\pm$ are orthogonal, they are perfectly distinguishable e.g.\ by means of photon counting; therefore, $\beta_1=1$. On the other hand, discriminating these two states by means of GOCC seems considerably more difficult. To arrive at a rigorous bound, we will leverage Sabapathy and Winter's estimate~\eqref{GOCC_vs_WL1}. But in order to do that we first need to derive an expression for the Wigner functions of the states $\omega_\lambda^\pm$. This is our immediate concern.

\begin{lemma} \label{omega_+-_Wigner_function}
For all $\lambda\in [0,1)$, the Wigner functions of the states $\omega_\lambda^+, \omega_\lambda^-$ defined by~\eqref{omega_+}--\eqref{omega_-} are given by
\begin{align}
W_{\!\omega_\lambda^+}(\alpha) &= \frac1\pi\left( (1\!-\!\lambda) e^{-2\frac{1-\lambda}{1+\lambda}|\alpha|^2} + (1\!+\!\lambda) e^{-2\frac{1+\lambda}{1-\lambda}|\alpha|^2} \right) , \\
W_{\!\omega_\lambda^-}(\alpha) &= \frac{1}{\pi\lambda} \left( (1\!-\!\lambda) e^{-2\frac{1-\lambda}{1+\lambda}|\alpha|^2} - (1\!+\!\lambda) e^{-2\frac{1+\lambda}{1-\lambda}|\alpha|^2} \right) .
\end{align}
\end{lemma}

\begin{proof}
The Wigner function of the $n^{\text{th}}$ Fock state is given by~\cite[Eq.~(4.5.31)]{BARNETT-RADMORE}
\bb
W_{\ket{n}\!\bra{n}}(\alpha) = \frac{2(-1)^n}{\pi}\, e^{-2|\alpha|^2}\, L_n(4|\alpha|^2)\, ,
\label{Fock_Wigner}
\ee
where $L_n$ is the $n^{\text{th}}$ Laguerre polynomial. We then obtain that
\bb
&W_{\omega_\lambda^+}(\alpha) \\
&\quad \texteq{1} (1-\lambda^2) \sum_{n=0}^\infty \lambda^{2n} W_{\ket{2n}\!\bra{2n}}(\alpha) \\
&\quad \texteq{2} \frac2\pi (1-\lambda^2)\, e^{-2 |\alpha|^2} \sum_{n=0}^\infty \lambda^{2n} L_{2n}(4|\alpha|^2) \\
&\quad = \frac1\pi (1-\lambda^2)\, e^{-2 |\alpha|^2} \bigg( \sum_{n=0}^\infty \lambda^{n} L_{n}(4|\alpha|^2) \\
&\hspace{24.5ex} + \sum_{n=0}^\infty (-\lambda)^{n} L_{n}(4|\alpha|^2) \bigg) \\
&\quad \texteq{3} \frac1\pi (1-\lambda^2)\, e^{-2 |\alpha|^2} \bigg( \frac{1}{1-\lambda}\, e^{-\frac{4\lambda}{1-\lambda}|\alpha|^2} \\
&\hspace{24ex} + \frac{1}{1+\lambda}\, e^{\frac{4\lambda}{1+\lambda}|\alpha|^2} \bigg) \\
&\quad = \frac1\pi \left( (1+\lambda)\, e^{-2\frac{1+\lambda}{1-\lambda}|\alpha|^2} + (1-\lambda)\, e^{-2\frac{1-\lambda}{1+\lambda}|\alpha|^2} \right) ,
\ee
where the identity in~1 follows from~\eqref{omega_+}, that in~2 from~\eqref{Fock_Wigner}, and that in~3 from~\cite[\S~22.9.15]{ABRAMOWITZ}. The Wigner function of the odd state $\omega_\lambda^-$ can be derived with a totally analogous computation, that we report for the sake of completeness. It holds that
\bb
&W_{\omega_\lambda^-}(\alpha) \\
&\quad = (1-\lambda^2) \sum_{n=0}^\infty \lambda^{2n} W_{\ket{2n+1}\!\bra{2n+1}}(\alpha) \\
&\quad = -\frac2\pi (1-\lambda^2)\, e^{-2 |\alpha|^2} \sum_{n=0}^\infty \lambda^{2n} L_{2n+1}(4|\alpha|^2) \\
&\quad = \frac{1}{\pi\lambda} (1-\lambda^2)\, e^{-2 |\alpha|^2} \bigg( \sum_{n=0}^\infty (-\lambda)^{n} L_{n}(4|\alpha|^2) \\
&\hspace{26ex} - \sum_{n=0}^\infty \lambda^{n} L_{n}(4|\alpha|^2) \bigg) \\
&\quad = \frac{1}{\pi\lambda} (1-\lambda^2)\, e^{-2 |\alpha|^2} \bigg( \frac{1}{1+\lambda}\, e^{\frac{4\lambda}{1+\lambda}|\alpha|^2} \\
&\hspace{26ex} - \frac{1}{1-\lambda}\, e^{-\frac{4\lambda}{1-\lambda}|\alpha|^2} \bigg) \\
&\quad = \frac{1}{\pi\lambda} \left( (1-\lambda)\, e^{-2\frac{1-\lambda}{1+\lambda}|\alpha|^2} - (1+\lambda)\, e^{-2\frac{1+\lambda}{1-\lambda}|\alpha|^2} \right) .
\ee
This concludes the proof.
\end{proof}

We are now ready to present the main result of this section, demonstrating the existence of data hiding against GOCC measurements even in the case of a single-mode system.

\begin{thm} \label{omega_+-_thm}
For $\lambda\in [0,1)$, the states $\omega_\lambda^+,\omega_\lambda^-$ defined by~\eqref{omega_+}--\eqref{omega_-} satisfy that
\bb
\frac12 \left\| \omega_\lambda^+ - \omega_\lambda^- \right\|_1 &= 1 \label{omega_+-_trace_norm}
\ee
but
\bb
\frac12 \left\| \omega_\lambda^+ - \omega_\lambda^- \right\|_\gocc &\leq \frac12 \left\| \omega_\lambda^+ - \omega_\lambda^- \right\|_{\Wp} \\
&\leq 2 \left( \frac{1-\lambda}{1+\lambda} \right)^{\frac{1+\lambda^2}{2\lambda}} \\
&= 1-\lambda + O\left( (1-\lambda)^2\right) \, . \label{omega_+-_wigd}
\ee
In particular, schemes of the form $(\omega_\lambda^+, \omega_\lambda^-, 1/2)$ exhibit data hiding against GOCC in the limit $\lambda\to 1^-$, i.e.
\bb
\beta_1\equiv 1 \quad \forall\ \lambda\qquad \text{but}\qquad \lim_{\lambda\to 1^-} \beta_\gocc = 0\, .
\label{GOCC_dh}
\ee
\end{thm}

\begin{proof}
Since $\omega_\lambda^+$ and $\omega_\lambda^-$ have orthogonal supports, they satisfy~\eqref{omega_+-_trace_norm} for all $\lambda$. As for~\eqref{omega_+-_wigd}, setting $x_+\coloneqq \max\{x,0\}$ write
\bb
&\frac12 \left\| \omega_\lambda^+ - \omega_\lambda^- \right\|_\gocc \\
&\quad \textleq{1} \frac12 \left\| \omega_\lambda^+ - \omega_\lambda^- \right\|_{\Wp} \\
&\quad \textleq{2} \frac12 \left\| W_{\omega_\lambda^+} - W_{\omega_\lambda^-} \right\|_{L^1} \\
&\quad \texteq{3} \int d^2\alpha\, \left( W_{\omega_\lambda^+}(\alpha) - W_{\omega_\lambda^-}(\alpha) \right)_{\!+} \\
&\quad \texteq{4} \frac{1}{\pi\lambda}\! \int\!\! d^2\alpha \left( \!(1\!+\!\lambda)^2 e^{-2\frac{1+\lambda}{1-\lambda} |\alpha|^2}\!\! - (1\!-\!\lambda)^2 e^{-2\frac{1-\lambda}{1+\lambda} |\alpha|^2} \right)_{\!+} \\
&\quad \texteq{5} \frac1\lambda\! \int_0^\infty\hspace{-2ex} du\, \left( (1+\lambda)^2 e^{-2\frac{1+\lambda}{1-\lambda} u} - (1-\lambda)^2 e^{-2\frac{1-\lambda}{1+\lambda} u} \right)_{\!+} \\
&\quad \texteq{6} \frac1\lambda\! \int_0^{\frac{1-\lambda^2}{4\lambda}\ln \frac{1+\lambda}{1-\lambda}} \hspace{-4ex} du \left(\! (1\!+\!\lambda)^2 e^{-2\frac{1+\lambda}{1-\lambda} u}\!\! - (1\!-\!\lambda)^2 e^{-2\frac{1-\lambda}{1+\lambda} u} \right)_{\!+} \\
&\quad = 2 \left( \frac{1+\lambda}{1-\lambda} \right)^{-\frac{1+\lambda^2}{2\lambda}} .
\label{GOCC_dh_proof_eq}
\ee
Here, in~1 and~2 we applied~\eqref{GOCC_vs_WL1}, in~3 we observed that both $W_{\omega_\lambda^+}$ and $W_{\omega_\lambda^-}$ integrate to $1$ thanks to~\eqref{integrate_Wigner}, in~4 we used Lemma~\ref{omega_+-_Wigner_function}, in~5 we performed the change of variable $u\coloneqq |\alpha|^2$, and in~6 we noted that
\bb
(1+\lambda)^2 e^{-2\frac{1+\lambda}{1-\lambda} u} &\geq (1-\lambda)^2 e^{-2\frac{1-\lambda}{1+\lambda} u}\\
\Longleftrightarrow\quad u&\leq \frac{1-\lambda^2}{4\lambda} \ln \frac{1+\lambda}{1-\lambda}\, .
\ee
This proves the first two inequalities of~\eqref{omega_+-_wigd}; Taylor expanding the rightmost side of~\eqref{GOCC_dh_proof_eq} yields also the last identity in~\eqref{omega_+-_wigd}. Finally,~\eqref{GOCC_dh} follows easily by combining~\eqref{omega_+-_trace_norm} and~\eqref{omega_+-_wigd}, once we remember that $\beta_\MM\left(\omega_\lambda^+,\omega_\lambda^-; 1/2\right)=\frac12 \left\|\omega_\lambda^+ - \omega_\lambda^-\right\|_\MM$ for $\MM=\mathrm{ALL}$ and $\MM=\gocc$.
\end{proof}

\tcr{We have thus demonstrated the sought data hiding against GOCC measurements. An analogous but substantially different result has recently been obtained by Sabapathy and Winter~\cite{KK-VV-GOCC, VV-GOCC-IHP, VV-GOCC-APS, VV-GOCC-ISIT} by means of a completely different construction. Their construction employs random convex mixtures of coherent states only, but has the drawback of requiring asymptotically many modes. Ours, on the contrary, requires a single mode only, but the two states to be prepared are highly non-classical. Because of these differences, depending on the platform chosen for the implementation one may be easier to realise or to break than the other, or vice versa.}

\subsection{Even and odd thermal states and non-ideal detection efficiency}

The even and odd thermal states constructed in~\eqref{omega_+}--\eqref{omega_-}, albeit almost indistinguishable under GOCC in the limit where $\lambda\to 1^-$, remain perfectly distinguishable for all $\lambda$, as expressed by~\eqref{omega_+-_trace_norm}. As we saw, there is a simple-to-describe experimental procedure to discriminate them with certainty. This involves making a photo-detection and counting the number of photons recorded: if this is an even number, then we know that the state was $\omega_\lambda^+$, otherwise it must have been $\omega_\lambda^-$.

In the experimental practice, measurements such as photo-detection are never ideal. Thus, we need to ask ourselves: \emph{what happens to the above discrimination protocol when the ideal photo-detectors are replaced by non-ideal ones?} The answer is contained in the forthcoming Proposition~\ref{final_prop}: the scheme with even and odd thermal states becomes once again data hiding, i.e.\ $\omega_\lambda^+$ and $\omega_\lambda^-$ become almost indistinguishable in the limit where $\lambda\to 1^-$.

To arrive at a precise quantitative estimate of this phenomenon, we need to first choose a figure of merit that captures the quality of an arbitrary measurement. A common way of doing that is to introduce a parameter $\eta\in [0,1)$ such that $1-\eta>0$ represents the probability that any single photon of the incoming signal is lost before being reckoned~\cite{Hogg2014}. Mathematically, the process by which each photon is subjected to a loss with probability $1-\eta$ (and such that each even is independent of each other) can be modelled by a \emph{pure loss channel}. This is a single-mode quantum channel, i.e.\ a completely positive trace preserving map on the space $\T(\HH_1)$ of trace class operators on $\HH_1=L^2(\R)$, whose action can be defined by
\bb
\EE_\eta(\rho) \coloneqq \Tr_2 \left[ U_\eta \left(\rho\otimes \ketbra{0}\right) U_\eta^\dag \right] .
\ee
Here, the partial trace is with respect to the second mode of the two-mode state $U_\eta (\rho\otimes \ketbra{0}) U_\eta^\dag$, and
\bb
U_\eta \coloneqq e^{\arccos\sqrt\eta \left(a^\dag b - a b^\dag \right)}
\ee
is the unitary representing the action of a beam splitter with transmissivity $\eta\in [0,1]$ on two modes with annihilation operators $a$ and $b$.

With this notation in place, we can follow~\cite{Hogg2014} (see also~\cite{Berry2010, Berry2011}) and define the set of measurements with efficiency $\eta\in [0,1]$ as those POVMs of the form $\EE_\eta^\dag\circ E$, where $E: \pazocal{A}\to \B_+(\HH)$ is an arbitrary POVM and $\EE_\eta^\dag$ is the adjoint of $\EE_\eta$ (see~\eqref{adjoint}). Although seemingly cumbersome, the POVM $\EE_\eta^\dag\circ E$ is very simply described in terms of the outcome obtained upon measuring a state $\rho$ (cf.~\eqref{outcome}). This is modelled by a random variable $X'$ with probability measure
\bb
\mu_{X'}(dx) = \Tr\left[ \rho \left( \EE_\eta^\dag\!\circ\! E \right) (dx) \right] = \Tr \left[ \EE_\eta(\rho)\, E(dx)\right] .
\ee
In what follows, we will denote with $\MM_\eta$ the set of single-mode measurements with efficiency (at most) $\eta$. We are now ready to prove the following addendum to Theorem~\ref{omega_+-_thm}.

\begin{prop} \label{final_prop}
For all $\eta\in [0,1]$, the two states $\omega_\lambda^\pm$ defined by~\eqref{omega_+}--\eqref{omega_-} satisfy
\bb
\frac12 \left\| \omega_\lambda^+\!\! - \omega_\lambda^- \right\|_{\MM_\eta} \!\leq \frac{\left( 1\!-\!\lambda\right) \left(\eta(1\!-\!\lambda) + (1\!-\!\eta)(1\!+\!\lambda) + 1\right)}{\left(1\!+\!\lambda\right) \left(\eta(1\!-\!\lambda) + (1\!-\!\eta) (1\!+\!\lambda)\right)}\, ,
\ee
where $\MM_\eta$ is the set of measurements with efficiency at most $\eta$. In particular, $\omega_\lambda^+$ and $\omega_\lambda^-$ are almost indistinguishable under measurements in $\MM_\eta$ if $\frac{1-\lambda}{1-\eta}\ll 1$, thus achieving data hiding against $\MM_\eta$ in the limit $\lambda\to 1^-$ for each fixed $\eta<1$.
\end{prop}

Proposition~\ref{final_prop} pinpoints the degree of accuracy that our photo-detectors should possess in order for the discrimination between $\omega_\lambda^+$ and $\omega_\lambda^-$ to be successful with high probability. The closer $\lambda$ is to $1$, the smaller $1-\eta$ needs to be. In particular, the scheme with even and odd thermal states is data hiding not only with respect to GOCC, but also with respect to any set of measurements with maximum efficiency $\eta<1$, irrespectively of whether they are Gaussian or not. We are now ready to provide a full proof of the above claim.

\begin{proof}[Proof of Proposition~\ref{final_prop}]
Thanks to~\eqref{mm_norm} and to the definition of efficiency, we have that
\bb
&\frac12 \left\| \omega_\lambda^+ - \omega_\lambda^- \right\|_{\MM_\eta} \\
&\qquad = \sup_{E'\in \MM_\eta} \frac12 \int_\pazocal{X} \left| \Tr\left[ \left( \omega_\lambda^+ - \omega_\lambda^- \right) E'(dx)\right]\right| \\
&\qquad = \sup_{E} \frac12 \int_\pazocal{X} \left| \Tr\left[ \left( \omega_\lambda^+ - \omega_\lambda^- \right) \left( \EE_\eta^\dag \circ E\right)(dx)\right]\right| \\
&\qquad = \sup_{E} \frac12 \int_\pazocal{X} \left| \Tr\left[ \EE_\eta \left( \omega_\lambda^+ - \omega_\lambda^- \right) E(dx)\right]\right| \\
&\qquad = \frac12 \left\| \EE_\eta \left( \omega_\lambda^+ - \omega_\lambda^- \right) \right\|_1\, ,
\ee
where the supremum in the second line is over an arbitrary POVM $E$, and last identity is just the Holevo--Helstrom theorem. It thus remains to estimate $\left\| \EE_\eta \left( \omega_\lambda^+ - \omega_\lambda^- \right) \right\|_1$. This can be swiftly accomplished by means of a trick.

Denoting with $a^\dag a$ the number operator, as usual, consider that
\bb
\frac{\omega_\lambda^+ - \lambda\omega_\lambda^-}{1+\lambda} &= (1-\lambda) \sum_{n=0}^\infty \lambda^{2n} \ketbra{2n} \\
&\hspace{3ex} - \lambda(1-\lambda)\sum_{n=0}^\infty \lambda^{2n} \ketbra{2n+1} \\
&= (1-\lambda)\, (-\lambda)^{a^\dag a}\, ,
\label{omega_thermal}
\ee
where of course $(-\lambda)^{a^\dag a} = \sum_{n=0}^\infty (-\lambda)^n \ketbra{n}$. The characteristic function~\eqref{chi} of the above operator is a (centred) Gaussian, because
\bb
\chi_{(1-t)\, t^{a^\dag a}}(\alpha) = e^{-\frac12 \frac{1+t}{1-t}\, |\alpha|^2}
\label{chi_t}
\ee
holds for all $t\in(-1,1)$. Since the pure loss channel acts on characteristic function as
\bb
\EE_\eta: \chi_\rho \longmapsto \chi_{\EE_\eta(\rho)}(\alpha) = \chi_\rho\left( \sqrt\eta \alpha\right) e^{-\frac12 (1-\eta) |\alpha|^2}\, ,
\label{E_eta_chi}
\ee
remembering that $\chi_{a \rho}(\alpha) = a\, \chi_\rho(\alpha)$ we deduce that
\bb
&\chi_{\EE_\eta\left(\omega_\lambda^+\! - \lambda \omega_\lambda^- \right)}(\alpha) \\
&\quad \texteq{\eqref{E_eta_chi}}\ \chi_{\omega_\lambda^+ - \lambda \omega_\lambda^-} \left( \sqrt\eta \alpha\right) e^{-\frac12 (1-\eta) |\alpha|^2} \\
&\quad \texteq{\eqref{omega_thermal}}\ (1\!-\!\lambda)\, \chi_{(1+\lambda)\,(-\lambda)^{a^\dag a}}\left( \sqrt\eta \alpha\right) e^{-\frac12 (1-\eta) |\alpha|^2} \\
&\quad \texteq{\eqref{chi_t}}\ (1\!-\!\lambda)\, \exp\left[-\frac{\eta}{2} \frac{1\!-\!\lambda}{1\!+\!\lambda} |\alpha|^2 -\frac12 (1\!-\!\eta) |\alpha|^2 \right] .
\ee
By setting
\bb
s\coloneqq -\frac{\lambda\eta}{1+\lambda (1-\eta)}
\ee
we can rewrite this as
\bb
\chi_{\EE_\eta\left(\omega_\lambda^+ - \lambda \omega_\lambda^- \right)}(\alpha) = (1-\lambda)\, e^{-\frac12 \frac{1+s}{1-s}\, |\alpha|^2}\, .
\ee
By comparing this with~\eqref{chi_t}, and remembering that the correspondence between trace class operators and characteristic functions is one-to-one, we deduce that
\bb
\EE_\eta\left(\omega_\lambda^+ - \lambda \omega_\lambda^- \right) = (1-\lambda)(1-s)\, s^{a^\dag a}\, .
\ee
Therefore, remembering that $\big\|s^{a^\dag a}\big\|_1 = \frac{1}{1-|s|} = \frac{1}{1+s}$, we infer
\bb
\left\| \frac{\omega_\lambda^+ - \lambda\omega_\lambda^-}{1+\lambda} \right\|_1 &= \left\| \frac{(1-\lambda)(1-s)}{1+\lambda}\, s^{a^\dag a} \right\|_1 \\
&= \frac{(1-\lambda)(1-s)}{(1+\lambda)(1+s)} \\
&= \frac{1-\lambda}{(1-\lambda) \eta + (1+\lambda)(1-\eta)}\, .
\label{Gaussian_estimate}
\ee
Finally, putting everything together, we conclude that
\bb
&\frac12 \left\| \omega_\lambda^+ - \omega_\lambda^- \right\|_1 \\
&\quad \leq\ \left\| \frac12 \omega_\lambda^+\! - \frac{1}{1+\lambda}\omega_\lambda^+ \right\|_1\! + \left\| \frac12 \omega_\lambda^- - \frac{\lambda}{1+\lambda} \omega_\lambda^- \right\|_1 \\
&\hspace{5.35ex} + \left\| \frac{\omega_\lambda^+ - \lambda\omega_\lambda^-}{1+\lambda} \right\|_1 \\
&\quad =\ \frac{1-\lambda}{1+\lambda} + \left\| \frac{\omega_\lambda^+ - \lambda\omega_\lambda^-}{1+\lambda} \right\|_1 \\
&\quad \textleq{\eqref{Gaussian_estimate}}\ \frac{1-\lambda}{1+\lambda} + \frac{1-\lambda}{(1-\lambda) \eta + (1+\lambda)(1-\eta)} \\
&\quad =\ \frac{\left( 1-\lambda\right) \left(\eta(1-\lambda) + (1-\eta)(1+\lambda) + 1\right)}{\left(1+\lambda\right) \left(\eta(1-\lambda) + (1-\eta) (1+\lambda)\right)}	\, ,
\ee
which completes the proof. 
\end{proof}

\section{Discussion, outlook, and open problems}

\tcr{In this paper we have examined many different ways in which one could implement the fundamental quantum information primitive of data hiding with continuous variable systems. We have started by developing a detailed quantitative analysis of the Braunstein--Kimble teleportation protocol. Our result (Theorem~\ref{BK_accuracy_thm}) is the first of its kind, to be best of our knowledge: it tells us how good our resources (e.g.\ two-mode squeezing) should be in order to carry out said teleportation protocol with a prescribed accuracy for all states up to a given average photon number.

We have used this result as a lever to derive Theorem~\ref{telep_multimode_thm} and Corollary~\ref{telep_thm}: there we determine --- up to a multiplicative constant --- the maximum effectiveness that continuous variable data hiding schemes against the set of local operations and classical communication can achieve for a given value of the local energy, i.e.\ photon number.

We have also considered data hiding against the set of measurements that can be implemented with Gaussian operations and classical computation (or feed-forward of measurement outcomes), a.k.a.\ GOCC. We have found a simple example of a family of two single-mode states (called even and odd thermal states) that: (i)~can be constructed from a two-mode squeezed vacuum state by measuring the photon number parity on one side; (ii)~can be discriminated perfectly by means of photon counting; but (iii)~exhibits data hiding against GOCC, i.e.\ cannot be resolved by means of those measurements. At the same time, we cannot help but be amazed at the surprising effectiveness of GOCC and even simple Gaussian measurements at several state discrimination tasks of practical interest: as we have proved analytically, thermal states and even consecutive Fock states are discriminated quite well throughout the whole parameter range.}


We have seen that in the special case of even and odd thermal states the limitations associated with GOCC could be overcome by including in our tool kit an additional measurement, namely, photo-detection. Any practically achievable photo-detection scheme, however, is bound to have a non-ideal efficiency, quantifiable by a parameter $0\leq \eta<1$, so that $1-\eta$ represents the probability that each incoming photon is lost~\cite{Hogg2014, Berry2010}. Remarkably, the scheme with even and odd thermal states turns out to achieve data hiding against any set of (possibly non-Gaussian) measurements with non-ideal efficiency $\eta<1$.

\tcr{Going forward, we would like to advertise a notable open problem that stems from the above discussion: 

\begin{problem*}
What is the operational power of the set of measurements obtained by combining GOCC protocols and (destructive) photon counting (we could call it GOPC)? Does it exhibit data hiding? 
\end{problem*}

The appeal of the above question rests upon the consideration that almost all practical implementations of quantum measurements on continuous variable systems are build in terms of GOPC protocols. Hence, constructing a data hiding scheme against GOPC would show the fundamental operational incompleteness (in a certain sense) of the very tools we use to do quantum information with continuous variables. On the contrary, proving that the set of GOPC measurements is sufficiently large that it does not admit data hiding would be remarkable in just the opposite way. We leave this question for future investigation.
}


We speculate that the data hiding schemes proposed here may be used, in a not-too-distant future, to benchmark the quality of remote devices operating with CV platforms. A benchmarking protocol could consist in a simple binary state discrimination task such as that described in Section~\ref{even_odd_thermal_ex}. A correct solution up to some small error tolerance invariably indicates that the device under examination can successfully conduct non-Gaussian and highly efficient measurements.

\section*{Acknowledgements.}
I am grateful to Krishna Kumar Sabapathy and Andreas Winter for spurring my interest in the topics treated here, and for sharing a preliminary version of their paper~\cite{KK-VV-GOCC}. I also thank an anonymous referee for suggesting an alternative way to obtain an inequality similar to (although less tight than) that in Theorem~\ref{telep_multimode_thm}. I acknowledge financial support from the Alexander von Humboldt Foundation. 


\appendix

\renewcommand{\theequation}{A\arabic{equation}}
\renewcommand{\thethm}{A\arabic{thm}}
\setcounter{equation}{0}
\setcounter{figure}{0}
\setcounter{table}{0}
\setcounter{section}{0}
\makeatletter

\section{Strong non-uniform convergence for the Braunstein--Kimble teleportation protocol} \label{strong_not_uniform_app}

In this appendix we prove that the quality of the resources to be employed in the Braunstein--Kimble protocol in order to achieve a prescribed final accuracy will in general depend on the state to be teleported. In other words, for any pair $(r,\eta)$, no matter how large $r$ or how close $\eta$ to $1$, one can find an input state $\rho_{RA}$ (no $R$ system needed) such that the right-hand side of~\eqref{BK_ancilla} is very different from $\rho_{RA}$ in trace norm. In mathematical terms, one says that \emph{the convergence of the right-hand side of~\eqref{BK_ancilla} to $\rho_{RA}$ as $r\to\infty$ and $\eta\to 1^-$ is strong but not uniform.} This has been recently clarified~\cite[Appendix~A]{PLOB} (see also~\cite[Section~II.B]{Mark-strong-uniform}). We now present a different proof of this fact, that has the advantage of not employing any ancillary system $R$, unlike those available prior to this work. We start by introducing the single-mode \emph{squeezed states}, defined by~\cite[Eq.~(3.7.5)]{BARNETT-RADMORE}
\bb
\ket{\zeta(r)} &\coloneqq S(r) \ket{0}\, ,\\
S(r) &\coloneqq \exp \left[ \frac{r}{2} \left( (a^\dag)^2 - a^2 \right) \right] .
\label{squeezed_states}
\ee
As we are about to show, if the input is a highly squeezed state, large values of $r$ and $\eta$ are needed to make the Braunstein--Kimble protocol approximate ideal teleportation.

\begin{lemma} \label{diamond_equals_2_lemma}
For all (fixed) $\lambda>0$, the squeezed states $\zeta(r)\coloneqq \ketbra{\zeta(r)}$ satisfy that
\bb
\lim_{r\to\infty} \left\| \left(\NN_\lambda - I\right) (\zeta(r)) \right\|_1 = 2\, .
\label{1-to-1_noise}
\ee
\end{lemma}

\begin{proof}
We write
\bb
&\left\| \left(\NN_\lambda - I\right) (\zeta(r)) \right\|_1 \\
&\qquad \textgeq{1} \Tr \left[ (\id- 2\zeta(r)) \left(\NN_\lambda - I\right) (\zeta(r)) \right] \\
&\qquad= 2 \Tr \left[ \zeta(r) \left(I -\NN_\lambda\right) (\zeta(r)) \right] \\
&\qquad\texteq{2} 2 - \frac{2}{\pi\lambda} \int d^2 \alpha\, e^{-\frac{|\alpha|^2}{\lambda}}\, \left| \braket{\zeta(r)|D(\alpha)|\zeta(r)} \right|^2 \\
&\qquad\texteq{3} 2 - \frac{2}{\pi\lambda} \int d^2 \alpha\, e^{-\frac{|\alpha|^2}{\lambda}} e^{ -\frac12 \left( e^{2r}\, \alpha_R^2 + e^{-2r}\, \alpha_I^2\right)} \\
&\qquad= 2 - \frac{2}{\sqrt{\left( 1+ \frac{\lambda}{2}\, e^{2r} \right)\left( 1+ \frac{\lambda}{2}\, e^{-2r} \right)}}\, .
\ee
Here, 1~holds because $\left\|\id-2\zeta(r)\right\|_\infty=1$, in~2 we used the representation~\eqref{noise}, and in~3 we employed a well-known expression for the overlap between displaced squeezed states~\cite[Eq.~(4.4.42)]{BARNETT-RADMORE}. Taking the limit for $r\to\infty$ confirms~\eqref{1-to-1_noise}.
\end{proof}

The above Lemma~\ref{diamond_equals_2_lemma} shows that
\bb
\left\| \NN_\lambda - I\right\|_\diamond = \left\| \NN_\lambda - I\right\|_{1\to 1} = 2\, ,
\ee
where $\|\cdot\|_\diamond$ denotes simply the (standard) \emph{diamond norm}, obtained by setting $H_A=0$ in~\eqref{EC_diamond}, and $\|\mathcal{L}\|_{1\to 1}\coloneqq \sup_{\|X\|_1\leq 1} \left\|\mathcal{L}(X)\right\|_1$. Note that $2$ is the maximum value that these norms can take on the difference of two quantum channels. Since uniform convergence of a sequence of maps $\mathcal{L}_n$ to $\mathcal{L}$ means that $\lim_{n\to\infty} \left\|\mathcal{L}_n-\mathcal{L}\right\|_\diamond=0$, Lemma~\ref{diamond_equals_2_lemma} proves that indeed the Gaussian noise channel $\NN_\lambda$ does not converge uniformly to the identity channel as $\lambda\to 0^+$ --- and so neither does the channel~\eqref{BK_ancilla}. This is by no means a peculiar feature of this particular channel, as even generic unitaries exhibit a similar behaviour~\cite[Proposition~2]{VV-diamond}. At the heart of this phenomenon is the fact that the underlying Hilbert space is infinite dimensional.

\section{Data hiding and inequivalence of norms} \label{inequivalence_norms_app}

Recently, the connection between the notion of data hiding and certain geometric properties of the associated norm was studied in great detail~\cite{VV-dh, ultimate}. To appreciate it, remember that in our infinite-dimensional vector spaces there can be many inequivalent norms. In this context, two norms $\|\cdot\|_{(1)}$ and $\|\cdot\|_{(2)}$ are said to be equivalent if there exist constants $0<c,C<\infty$ such that
\bb
c\, \|\cdot\|_{(1)} \leq \|\cdot\|_{(2)} \leq C\|\cdot\|_{(1)}\, ,
\ee
meaning that said inequalities have to hold when the norms are evaluated on all elements of the underlying space, with $c,C$ independent of the chosen vector. While all norms are equivalent in finite-dimensional spaces, it is elementary to see that this it no longer true in infinite dimension. Moreover, and perhaps surprisingly, the inequivalence between the norms $\|\cdot\|_1$ and $\|\cdot\|_\MM$, i.e.\ the fact that there is no constant $C$ such that $\|\cdot\|_1\leq C\|\cdot\|_\MM$ --- the converse inequality $\|\cdot\|_\MM\leq \|\cdot\|_1$ always holds by definition --- has important physical consequences.

\begin{lemma}[{\cite{VV-dh, ultimate}}]
An informationally complete set of measurements $\MM$ on a system with Hilbert space $\HH$ exhibits data hiding if and only if the two norms $\|\cdot\|_1$ and $\|\cdot\|_\MM$ on the space $\T(\HH)$ of trace class operators on $\HH$ are not equivalent.
\end{lemma}

\begin{proof}
We report the following proof for the sake of completeness, although it follows easily from previous arguments~\cite{VV-dh, ultimate, XOR}. Let $\MM$ exhibit data hiding, and assume that $\|\cdot\|_1\leq C\|\cdot\|_\MM$ for some constant $C<\infty$. Then we can find schemes $(\rho_n,\sigma_n; p_n)_{n\in \N}$ satisfying~\eqref{dh}, and setting $X_n \coloneqq p_n \rho_n - (1-p_n)\sigma_n$ we see that 
\bbb
C \geq \limsup_{n\to\infty} \frac{\|X_n\|_1}{\|X_n\|_\MM} = +\infty\, ,
\eee
which leads to a contradiction.

Conversely, let the norms $\|\cdot\|_1$ and $\|\cdot\|_\MM$ be inequivalent. Then we can construct a sequence of trace class operators $Z_n\in \T(\HH)$ such that $\|Z_n\|_1\equiv 1$ but $\lim_{n\to\infty} \|Z_n\|_\MM=0$, where the first identity can always be enforced up to re-scaling $Z_n$ by an appropriate constant. We can also decompose $Z_n=X_n+iY_n$ into its self-adjoint and anti-self-adjoint parts. Since $1=\|Z_n\|_1\leq \|X_n\|_1 + \|Y_n\|_1$, we deduce that either $\min\left\{\|X_n\|_1,\, \|Y_n\|_1\right\}\geq 1/2$ for all $n$. Up to multiplying $Z_n$ by the imaginary unit $i$, we can assume that $\|X_n\|_1\geq 1/2$ for all $n$. Now, consider that 
\bb
\|X_n\|_\MM &= \left\| \frac{Z_n + Z^\dag_n}{2} \right\|_\MM \\
&\leq \frac12 \left( \|Z_n\|_\MM + \|Z_n^\dag \|_\MM \right) \\
&= \|Z_n\|_\MM\, ,
\ee
where the last equality follows from the elementarily verified fact that $\|Z\|_\MM= \|Z^\dag\|_\MM$ for all $Z$.

Set $X'_n\coloneqq \frac{X_n}{\|X_n\|_1}$. Since $X'_n$ has trace norm $1$, we can decompose it as $X'_n= p_n \rho_n - (1-p_n) \sigma_n$, where $p_n\in [0,1]$ and $\rho_n,\sigma_n$ are density operators with orthogonal support. On the one hand, this latter fact entails that $\beta_1\left(\rho_n, \sigma_n; p_n\right)\equiv 1$. On the other,
\bb
0 &\leq \limsup_{n\to\infty} \beta_\MM \left(\rho_n, \sigma_n; p_n\right) \\
&= \limsup_{n\to\infty} \|X'_n\|_\MM \\
&= \limsup_{n\to\infty} \frac{\|X_n\|_\MM}{\|X_n\|_1} \\
&\leq 2 \limsup_{n\to\infty} \frac{\|Z_n\|_\MM}{\|Z_n\|_1} \\
&= 0\, ,
\ee
completing the verification of conditions~\eqref{dh}.
\end{proof}

Leveraging the above connection with the geometric properties of norms, in~\cite[Appendix~A]{ultimate} it was also established that it does not make much difference to the existence of data hiding if we consider all possible schemes or only \emph{equiprobable} ones. We provide a simple self-contained proof of this elementary fact below.

\begin{lemma} \label{equiprobable_lemma}
An informationally complete set of measurements $\MM$ exhibits data hiding if and only if there is a sequence $\left( \rho_n, \sigma_n; 1/2\right)_{n\in \N}$ of equiprobable schemes that satisfies conditions~\eqref{dh}.
\end{lemma}

\begin{proof}
Since sufficiency is obvious, we prove necessity. Let us assume that there is a sequence of not necessarily equiprobable schemes $\left( \rho_n, \sigma_n; p_n\right)_{n\in \N}$ that exhibits data hiding against $\MM$, i.e.\ satisfies~\eqref{dh}. Then, we claim that $\lim_{n\to\infty} p_n = 1/2$. Indeed, using the trivial fact that $\|Z\|_\MM \geq \left| \Tr Z\right|$ for all trace class $Z$, as follows immediately by applying the convexity of the modulus function to~\eqref{mm_norm}, we see that
\bb
0 &= \lim_{n\to\infty} \beta_\MM \left( \rho_n, \sigma_n; p_n\right) \\
&= \lim_{n\to\infty} \left\|p_n \rho_n - (1-p_n)\sigma_n \right\|_\MM \\
&\geq \lim_{n\to\infty} \left| 2p_n - 1\right| ,
\ee
implying the claim.

Now, consider the equiprobable schemes $\left( \rho_n, \sigma_n; 1/2\right)_{n\in \N}$. On the one hand, leveraging the triangle inequality for the trace norm we obtain that
\bb
&\beta_1 \left( \rho_n, \sigma_n; 1/2\right) \\
&\quad = \frac12 \left\| \rho_n - \sigma_n \right\|_1 \\
&\quad \geq \left\| p_n \rho_n - (1-p_n) \sigma_n \right\|_1 - \left\| \left(p_n-\frac12\right) \rho_n \right\|_1 \\
&\quad \qquad - \left\| \left(p_n-\frac12\right) \sigma_n \right\|_1 \\
&\quad = \left\| p_n \rho_n - (1-p_n) \sigma_n \right\|_1 - \left| 2p_n - 1\right| ,
\ee
implying that $\lim_{n\to\infty} \beta_1 \left( \rho_n, \sigma_n; 1/2\right) = 1$. On the other, with the same method we obtain that
\bb
\beta_\MM \left( \rho_n, \sigma_n; 1/2\right) &= \frac12 \left\| \rho_n - \sigma_n \right\|_\MM \\
&\leq \left\| p_n \rho_n - (1-p_n) \sigma_n \right\|_\MM + \left| 2p_n - 1\right| ,
\ee
entailing that $\lim_{n\to\infty} \beta_\MM \left( \rho_n, \sigma_n; 1/2\right) = 0$. Thus, the sequence of schemes $\left( \rho_n, \sigma_n; 1/2\right)_{n\in \N}$ exhibits data hiding as well.
\end{proof}

\section{Proof of Lemma~\ref{no_dh_thermal_lemma} on the discrimination of thermal states} \label{thermal_app}

This appendix is devoted to a detailed presentation of the calculations leading to the expressions~\eqref{td_thermal}--\eqref{homd thermal} for the different biases of thermal-state-based schemes with respect to different sets of measurements.

\begin{proof}[Proof of Lemma~\ref{no_dh_thermal_lemma}]
Let us start by evaluating the trace norm distance $\left\|\tau_\nu - \tau_\mu\right\|_1$. Since $\tau_\nu$ and $\tau_\mu$ are both diagonal in the Fock basis, we have that
\bb
\frac12 \left\|\tau_\nu\! -\! \tau_\mu\right\|_1 &= \frac12 \sum_{k=0}^\infty \left| \braket{k|\tau_\nu|k} - \braket{k|\tau_\mu|k} \right| \\
&\texteq{1} \sum_{k=0}^\infty \max\left\{ \braket{k|\tau_\nu|k}- \braket{k|\tau_\mu|k},\, 0 \right\} \\
&= \sum_{k=0}^\infty \max\left\{ \frac{\nu^k}{(\nu+1)^{k+1}} - \frac{\mu^k}{(\mu+1)^{k+1}},\, 0 \right\} \\
&\texteq{2} \sum_{k=0}^{N_0(\nu,\mu)} \left( \frac{\nu^k}{(\nu+1)^{k+1}} - \frac{\mu^k}{(\mu+1)^{k+1}} \right) \\
&\texteq{3} \left(\frac{\mu}{\mu+1}\right)^{N_{0}(\nu,\mu)+1} - \left(\frac{\nu}{\nu+1}\right)^{N_{0}(\nu,\mu)+1}\!\!\! .
\ee
Here, in~1 we performed elementary manipulations and observed that $\Tr\tau_\nu=\Tr\tau_\mu=1$, in~2 we noted that
\bb
\frac{\nu^k}{(\nu+1)^{k+1}} \geq \frac{\mu^k}{(\mu+1)^{k+1}} \ &\Longleftrightarrow\ \left(\frac{\mu(\nu+1)}{\nu(\mu+1)} \right)^k\! \leq \frac{\mu+1}{\nu+1} \\[1ex]
&\Longleftrightarrow\ 0\leq k\leq N_0(\nu,\mu)\, ,
\ee
where $N_0(\nu,\mu)$ is given by~\eqref{N_0},
while in~3 we evaluated the geometric sum. Along the same lines, one also shows that
\bb
\frac12 \left\|\tau_\nu\! -\! \tau_\mu\right\|_1 = \max_{N\in \N} \left(\!\! \left(\! \frac{\mu}{\mu\!+\!1}\! \right)^{\!\!N+1} - \left(\! \frac{\nu}{\nu\!+\!1}\! \right)^{\!\!N+1} \right) . \label{no_dh_thermal_proof_eq1}
\ee
This completes the proof of~\eqref{td_thermal}. We now move on to~\eqref{wigd_thermal}. We have that
\bb
&\frac12 \left\|\tau_\nu - \tau_\mu\right\|_{\Wp} \\
&\quad\textleq{4} \frac12 \left\| W_{\tau_\nu} - W_{\tau_\mu} \right\|_{L^1} \\
&\quad\texteq{5} \int \!d^2\alpha\, \max\left\{ W_{\tau_\nu}(\alpha) - W_{\tau_\mu}(\alpha),\, 0 \right\} \\
&\quad\texteq{6} \frac{2}{\pi}\! \int \!d^2\alpha\, \max\left\{\frac{1}{2\nu\!+\!1}\, e^{-\frac{2|\alpha|^2}{2\nu+1}} - \frac{1}{2\mu\!+\!1}\, e^{-\frac{2|\alpha|^2}{2\mu+1}} ,\, 0\right\} \\
&\quad\texteq{7} 2 \!\int_0^\infty \hspace{-2ex} ds\, \max\left\{ \frac{1}{2\nu\!+\!1}\, e^{-\frac{2s}{2\nu+1}} - \frac{1}{2\mu\!+\!1}\, e^{-\frac{2s}{2\mu+1}},\, 0 \right\} \\
&\quad\texteq{8} 2 \int_0^{s_0(\nu,\mu)} ds\, \left( \frac{1}{2\nu\!+\!1}\, e^{-\frac{2s}{2\nu+1}} - \frac{1}{2\mu\!+\!1}\, e^{-\frac{2s}{2\mu+1}} \right) \\
&\quad= e^{-\frac{2s_0(\nu,\mu)}{\mu+1}} - e^{-\frac{2s_0(\nu,\mu)}{\nu+1}} \\
&\quad= \frac{2(\mu-\nu)}{2\mu+1} \left(\frac{2\mu+1}{2\nu+1} \right)^{-\frac{2\nu+1}{2(\mu-\nu)}} .
\ee
The above derivation can be justified as follows: in~4 we applied Sabapathy and Winter's estimate~\eqref{GOCC_vs_WL1}; 5~follows by elementary manipulations, leveraging the fact that $W_{\tau_\nu}$ and $W_{\tau_\mu}$ integrate to $1$ because of~\eqref{integrate_Wigner}; in~6 we used~\eqref{Wigner_thermal}; in~7 we performed the change of variables $s\coloneqq |\alpha|^2$; finally, in~8 we noted that
\bb
\frac{1}{2\nu+1} e^{-\frac{2s}{2\nu+1}} \geq \frac{1}{2\mu+1} e^{-\frac{2s}{2\mu+1}} \quad \Longleftrightarrow\quad 0\leq s\leq s_0(\nu,\mu) \, ,
\ee
where
\bb
s_0(\nu,\mu) \coloneqq \frac{(2\nu+1)(2\mu+1)}{4(\mu-\nu)}\, \ln\left( \frac{2\mu+1}{2\nu+1} \right) .
\ee
This concludes the proof of~\eqref{wigd_thermal}. We now move on to~\eqref{hetd_thermal}. We write that
\bb
&\frac12 \left\|\tau_\nu - \tau_\mu\right\|_{\mathrm{het}} \\
&\quad = \frac12 \left\| Q_{\tau_\nu} - Q_{\tau_\mu} \right\|_{L^1} \\
&\quad \texteq{9} \int d^2\alpha \max\left\{ Q_{\tau_\nu}(\alpha) - Q_{\tau_\mu}(\alpha),\, 0 \right\} \\
&\quad = \frac1\pi \int d^2\alpha \max\left\{ \frac{1}{\nu+1} \, e^{-\frac{|\alpha|^2}{\nu+1}} - \frac{1}{\mu+1} \, e^{-\frac{|\alpha|^2}{\mu+1}},\, 0 \right\} \\
&\quad \texteq{10} \int_0^\infty dt\, \max\left\{ \frac{1}{\nu+1} \, e^{-\frac{t}{\nu+1}} - \frac{1}{\mu+1} \, e^{-\frac{t}{\mu+1}},\, 0 \right\} \\
&\quad \texteq{11} \int_0^{t_0(\nu,\mu)} dt\, \left( \frac{1}{\nu+1} \, e^{-\frac{t}{\nu+1}} - \frac{1}{\mu+1} \, e^{-\frac{t}{\mu+1}} \right) \\
&\quad = e^{-\frac{t_0(\nu,\mu)}{\mu+1}} - e^{-\frac{t_0(\nu,\mu)}{\nu+1}} \\
&\quad = \frac{\mu-\nu}{\mu+1} \left(\frac{\mu+1}{\nu+1}\right)^{-\frac{\nu+1}{\mu-\nu}} .
\ee
Here, 9~holds once again because $Q_{\tau_\nu}$ and $Q_{\tau_\mu}$ are probability distributions, in~10 we performed the change of variables $t\coloneqq |\alpha|^2$, and in~11 we observed that
\bb
\frac{1}{\nu+1} e^{-\frac{t}{\nu+1}} \geq \frac{1}{\mu+1} e^{-\frac{t}{\mu+1}} \quad \Longleftrightarrow\quad 0\leq t\leq t_0(\nu,\mu) \, ,
\ee
where
\bb
t_0(\nu,\mu) \coloneqq \frac{(\mu+1)(\nu+1)}{\mu-\nu}\, \ln\left( \frac{\mu+1}{\nu+1} \right) .
\ee
This completes also the proof of~\eqref{hetd_thermal}. 

As for~\eqref{hetd_thermal}, remember that a homodyne detection is the measurement of a canonical quadrature. Since both $\tau_\nu$ and $\tau_\mu$ are invariant under phase space rotations, we can assume without loss of generality that the measured quadrature is $x=(a+a^\dag)/\sqrt2$. Since the wave function of the $n^{\text{th}}$ Fock state $\ket{n}$ is well known to be
\bb
\psi_n(x) = \frac{1}{\pi^{1/4} \sqrt{2^n n!}}\, H_n(x)\, e^{-x^2/2}\, ,
\label{Fock_wave_function}
\ee
the probability distribution of the outcome of our homodyne detection reads
\bb
P^{\mathrm{hom}}_{\tau_\nu}(x) &= \frac{1}{\nu+1} \sum_{k=0}^\infty \left( \frac{\nu}{\nu+1} \right)^k \left| \psi_k(x) \right|^2 \\
&= \frac{e^{-x^2}}{\sqrt\pi (\nu+1)} \sum_{k=0}^\infty \frac{1}{2^k k!} \left( \frac{\nu}{\nu+1} \right)^k H_k(x)^2 \\
&\texteq{12} \frac{1}{\sqrt{\pi (2\nu+1)}}\, e^{-\frac{x^2}{2\nu+1}} ,
\ee
where in~12 we applied Mehler's formula~\cite[p.~194]{TRANSCENDENTAL-2}. We deduce that
\bb
&\frac12 \left\| \tau_\nu - \tau_\mu\right\|_{\mathrm{hom}} \\
&\quad = \int_{-\infty}^{+\infty} dx\, \max\left\{ P^{\mathrm{hom}}_{\tau_\nu}(x) - P^{\mathrm{hom}}_{\tau_\mu}(x),\, 0 \right\} \\
&\quad = \frac{1}{\sqrt\pi} \int_{-\infty}^{+\infty} dx\, \max\left\{ \frac{e^{-\frac{x^2}{2\nu+1}}}{\sqrt{2\nu+1}} - \frac{e^{-\frac{x^2}{2\mu+1}}}{\sqrt{2\mu+1}} ,\, 0 \right\} \\
&\quad \texteq{13} \frac{1}{\sqrt\pi} \int_{-x_0(\nu,\mu)}^{+x_0(\nu,\mu)} dx\, \left( \frac{e^{-\frac{x^2}{2\nu+1}}}{\sqrt{2\nu+1}} - \frac{e^{-\frac{x^2}{2\mu+1}}}{\sqrt{2\mu+1}} \right) \\
&\quad = \erf\left(\frac{x_0(\nu,\mu)}{\sqrt{2\nu+1}} \right) - \erf\left(\frac{x_0(\nu,\mu)}{\sqrt{2\mu+1}} \right) \\
&\quad = \erf\left(\sqrt{\frac{2\mu+1}{4(\mu-\nu)}\, \ln\left( \frac{2\mu+1}{2\nu+1} \right)} \right) \\
&\hspace{5ex} - \erf\left(\sqrt{\frac{2\nu+1}{4(\mu-\nu)}\, \ln\left( \frac{2\mu+1}{2\nu+1} \right)} \right) ,
\ee
where in~13 we noted that
\bb
\frac{e^{-\frac{x^2}{2\nu+1}}}{\sqrt{2\nu+1}} \geq \frac{e^{-\frac{x^2}{2\mu+1}}}{\sqrt{2\mu+1}} \quad \Longleftrightarrow\quad |x|\leq x_0(\nu,\mu) \, ,
\ee
where
\bb
x_0(\nu,\mu) \coloneqq \sqrt{\frac{(2\nu+1)(2\mu+1)}{4(\mu-\nu)}\, \ln\left( \frac{2\mu+1}{2\nu+1} \right)}\, .
\ee
This concludes also the proof of~\eqref{homd thermal}.
\end{proof}

\section{Proof of Proposition~\ref{no_dh_Fock_prop} on discrimination of consecutive Fock states via homodyning} \label{Fock_app}

The purpose of this appendix is to present a complete and fully rigorous proof of Proposition~\ref{no_dh_Fock_prop}. Before we do that, we need to carry out a preliminary computation.

\begin{lemma} \label{palla_lemma}
Let $\epsilon>0$ be fixed. Then for all $\varphi\in [\epsilon, \pi-\epsilon]$ we have that
\bb
&\int_\epsilon^\varphi d\theta\, \left|\cos\left( 2(n+1)\theta - \left( n+\frac12 \right) \sin(2\theta) \right)\right| \\
&\qquad = \frac2\pi(\varphi-\epsilon) + O (n^{-1})\, .
\ee
as $n\to \infty$, where the estimate on the right-hand side holds uniformly in $\varphi$.
\end{lemma}

\begin{proof}
Hereafter we take $\epsilon$ to be a non-zero number that is fixed once and for all. Define the sequence of functions $f_n:[\epsilon, \pi-\epsilon] \to \R$ by $f_n(\varphi)\coloneqq 2 \varphi \left( 1+ \frac{1}{2n+1}\right) - \sin(2\varphi)$. It is easy to verify that $f_n'(\varphi) = 4\sin^2(\varphi) + \frac{2}{2n+1} \geq 4 \sin^2(\epsilon)$, from which it follows that (i)~$f_n$ is strictly increasing, that (ii)~$f_n'$ and $1/f'_n$ are both uniformly bounded; and finally that (iii)~$|f_n''(\varphi)| = 4 |\sin(2\varphi)|\leq 4$ is uniformly bounded as well. In light of~(i) we can consider the inverse function $g_n: [f_n(\epsilon), f_n(\pi-\epsilon)] \to [\epsilon, \pi-\epsilon]$. Thanks to~(ii) we have that (iv)~the derivative $g'_n$ is uniformly bounded. Moreover, one sees that 
\bb
\left| g_n''\left(f_n(\varphi)\right) \right| &= \left| \frac{f_n''(\varphi)}{f_n'(\varphi)^2} \right| \\
&= \left| \frac{\sin(2\varphi)}{\left( 2\sin^2\varphi + \frac{1}{2n+1} \right)^2} \right| \leq \frac{1}{2\sin^3\epsilon}\, ,
\ee
i.e.\ (v)~$g_n''$ is also uniformly bounded. Now, set
\begin{align*}
k_{\min}(n) \coloneqq&\ \min\!\left\{ k\!\geq\! 0: k \pi \geq\! \left(n\!+\!\frac12\right) f_n(\epsilon) \right\} \\
=&\ \ceil{\frac{1}{\pi} \left(n\!+\!\frac12\right) f_n(\epsilon)} , \\[1ex]
k_{\max}(n,\varphi) \coloneqq&\ \max\!\left\{ k\!\geq\! 0: (k\!+\!1) \pi \leq\! \left(n\!+\!\frac12\right) f_n(\varphi) \right\} \\
=&\ \floor{\frac{1}{\pi} \left(n\!+\!\frac12\right) f_n(\varphi) - 1} .
\end{align*}
For $k=k_{\min}(n), \ldots, k_{\max}(n,\varphi)$, call $\theta_{n,k}\coloneqq g_n \left( \frac{2k\pi}{2n+1} \right)$. Since
\bbb
k_{\min}(n) &\leq \frac{1}{\pi} \left(n+\frac12\right) f_n(\epsilon) + 1\, ,\\
k_{\max}(n,\varphi) &\geq \frac{1}{\pi} \left(n+\frac12\right) f_n(\varphi) - 1\, ,
\eee
and $g_n$ is monotonic (thanks to~(i)) with bounded derivative (see~(iv)), we have that
\bb
\epsilon&\leq \theta_{n,\, k_{\min}(n)} \leq \epsilon + O(n^{-1})\, ,\\
\varphi + O(n^{-1}) &\leq \theta_{n,\, k_{\max}(n,\varphi) + 1} \leq \varphi\, .
\label{approximate_interval}
\ee
Performing a Taylor expansion with Lagrange remainder term and leveraging property~(v) above, we see that
\bb
\begin{aligned}
\theta_{n,k+1} - \theta_{n,k} &= g_n \left( \frac{2k\pi}{2n\!+\!1} + \frac{2\pi}{2n\!+\!1} \right) - g_n \left( \frac{2k\pi}{2n\!+\!1} \right) \\
&= \frac{2\pi}{2n\!+\!1}\, g_n'\left( \frac{2k\pi}{2n\!+\!1} \right) + O\left(n^{-2}\right) .
\end{aligned}
\label{difference_roots}
\ee
uniformly for all $k$. Putting all together,
\begin{align*}
&\int_\epsilon^\varphi d\theta\, \left|\cos\left( 2(n+1)\theta - \left( n+\frac12 \right) \sin(2\varphi) \right)\right| \\
&\quad = \int_\epsilon^\varphi d\theta\, \left|\cos\left( \left( n+\frac12 \right) f_n(\theta) \right)\right| \\
&\quad \texteq{1} \sum_{k=k_{\min}(n)}^{k_{\max}(n,\varphi)} \!\int_{\theta_{n,k}}^{\theta_{n,k+1}} \hspace{-2.5ex} d\theta\, \left|\cos\!\left(\! \left( n\!+\!\frac12 \right) f_n(\theta)\! \right)\right| + O\!\left(\frac1n\right) \\
&\quad \texteq{2} \frac{2}{2n\!+\!1} \!\sum_{k=k_{\min}(n)}^{k_{\max}(n,\varphi)}\! \int_{k \pi}^{(k+1) \pi} \hspace{-3.5ex} d\alpha\ g'_n\!\! \left( \!\frac{2\alpha}{2n\!+\!1}\!\right) \!\left|\cos \alpha \right| + O\!\left(\frac1n\right) \\
&\quad \texteq{3} \frac{2}{2n\!+\!1} \sum_{k=k_{\min}(n)}^{k_{\max}(n,\varphi)} \int_{k \pi}^{(k+1) \pi} \hspace{-3.5ex} d\alpha\, \left( g'_n\! \left( \frac{2k\pi}{2n\!+\!1}\right) + O\!\left(\frac1n\right)\! \right) \\
&\hspace{29.6ex} \times \left|\cos \alpha \right| + O\!\left(\frac1n\right) \\
&\quad \texteq{4} \frac{2}{2n\!+\!1} \sum_{k=k_{\min}(n)}^{k_{\max}(n,\varphi)}  2\, g'_n\! \left( \frac{2k\pi}{2n\!+\!1}\right) + O\!\left(\frac1n\right) \\
&\quad \texteq{5} \frac2\pi \sum_{k=k_{\min}(n)}^{k_{\max}(n,\varphi)} \left( \theta_{n,k+1} - \theta_{n,k}\right)  + O\!\left(\frac1n\right) \\
&\quad = \frac2\pi \left( \theta_{n,\,k_{\max}(n,\varphi)+1} - \theta_{n,\,k_{\min}(n)} \right) + O\!\left(\frac1n\right) \\
&\quad \texteq{6} \frac2\pi \left( \varphi - \epsilon \right) + O\!\left(\frac1n\right)
\end{align*}
The justification of the above derivation is as follows. In~1 we used~\eqref{approximate_interval}; in~2 we performed the change of variable $\alpha = \left( n+\frac12\right) f_n(\theta)$, i.e.\ $\theta = g_n\left( \frac{2\alpha}{2n+1}\right)$; in~3 we leveraged~(v), i.e.\ the uniform boundedness of $g_n''$; in~4 we computed the integral over $\alpha$, which is elementary in light of the fact that $\cos\alpha$ does not change sign in the specified interval; in~5 we employed~\eqref{difference_roots}; finally, in~6 we exploited once again~\eqref{approximate_interval}.
\end{proof}

We are finally ready to assess the discriminability of two consecutive Fock states via homodyning.

\begin{proof}[Proof of Proposition~\ref{no_dh_Fock_prop}]
The quantity we have to compute takes the form
\bb
\left\|\ketbra{n} - \ketbra{n\!+\!1}\right\|_{\mathrm{hom}} &= \left\| \psi_n^2 - \psi_{n+1}^2 \right\|_{L^1} \\
&= \int_{-\infty}^{+\infty} \hspace{-2.5ex} dx \left|\psi_n^2(x) - \psi_{n+1}^2(x)\right| ,
\label{homd_Fock}
\ee
where the wave function of the $n^{\text{th}}$ Fock state is given by~\eqref{Fock_wave_function}. We will now be concerned with the asymptotic expansions of this expression as $n\to\infty$. Consider a constant $\epsilon>0$, fixed from now on unless otherwise specified, and set $x=\sqrt{2n+1}\cos\varphi$, where $\epsilon< \varphi< \pi-\epsilon$. Then it is known that~\cite[Theorem~8.22.9]{SZEGO} (see also the original work by Plancherel and Rotach~\cite{Plancherel1929}, as well as the paper by Tricomi~\cite{Tricomi1941})
\bb
&e^{-x^2/2} H_n(x) \\
&\quad = \frac{2^{\frac{n}{2}+\frac14}\sqrt{n!}}{(\pi n)^{1/4}}\, \frac{1}{\sqrt{\sin\varphi}} \\
&\qquad \times \left( \sin\!\left(\! \left(\frac{n}{2}\! +\! \frac14\right)(2\varphi-\sin(2\varphi)) + \frac{\pi}{4} \right) + O\!\left(\frac1n\right)\!\right) ,
\label{asymptotics_Hermite}
\ee
with the approximation holding uniformly for all $\varphi\in [\epsilon, \pi-\epsilon]$. Plugging this into~\eqref{Fock_wave_function} yields
\bb
\psi_n(x) &= \frac{2^{1/4}}{\sqrt\pi\, n^{1/4}} \, \frac{1}{\sqrt{\sin\varphi}} \\
&\quad \times \bigg(\! \sin\!\left(\! \left(\frac{n}{2}\! +\! \frac14\right)\!(2\varphi-\sin(2\varphi)) + \frac{\pi}{4} \right) \\
&\hspace{6.8ex} + O \!\left(\frac1n\right)\!\bigg)
\ee
and in turn, upon elementary algebraic manipulations,
\bb
\psi_n^2(x) &= \frac{1}{\pi\sqrt{2n}}\, \frac{1}{\sin\varphi} \\
&\quad \times \bigg( 1 + \sin\!\left(\! \!\left(n \!+\! \frac12\right)\!(2\varphi-\sin(2\varphi))\! \right) \\
&\hspace{6ex} + O \!\left(\frac1n\right)\!\bigg) .
\ee

To derive the expression for the integrand on the right-hand side of~\eqref{homd_Fock}, set as usual $\varphi = \arccos\left( \frac{x}{\sqrt{2n+1}}\right)$, and define the corresponding variable for $n\to n+1$ as
\bb
\varphi' \coloneqq&\ \arccos \left( \frac{x}{\sqrt{2n+3}}\right) \\
=&\ \varphi + \frac{\cot \varphi}{2n+1} + O\!\left(\frac{1}{n^2}\right) \\
=&\ \varphi + O\!\left(\frac1n\right) .
\ee
We have that
\begin{widetext}
\bb
&\pi \sqrt{\frac{n}{2}}\left( \psi_n^2(x) - \psi^2_{n+1}(x) \right) \\
&\quad \texteq{1} \frac{1}{2 \sin\varphi} \left( 1 + \sin\left( \left(n + \frac12\right)(2\varphi-\sin(2\varphi)) \right) + O\! \left(n^{-1}\right)\right) \\
&\quad \qquad - \frac{1 + O\!\left(n^{-1}\right)}{2\sin\varphi + O\!\left(n^{-1}\right)} \left( 1 + \sin\left( \left(n + \frac32\right)\left(2\varphi-\sin(2\varphi) + \frac{2\sin(2\varphi)}{2n+1} \right) \right) + O\!\left(\frac1n\right)\right) \\
&\quad = \frac{1}{2\sin\varphi} \sin\left( \left(n + \frac12\right)(2\varphi-\sin(2\varphi)) \right) \\
&\quad \qquad -\frac{1}{2\sin\varphi}\, \sin\left( \left(n + \frac32\right)\left(2\varphi - \frac{2n-1}{2n+1} \sin(2\varphi) \right) \right) + O\!\left(\frac1n\right) \\
&\quad \texteq{2} \frac{1}{\sin\varphi}\, \sin\left(\! \frac12\! \left(\! -2\varphi - \frac{2 \sin(2\varphi)}{2n\!+\!1}\right)\!\! \right) \cos\left(\!\frac12\! \left(\! (2n\!+\!2) 2\varphi - \left( \!2n\!+\!1 - \frac{2}{2n\!+\!1}\right) \sin(2\varphi) \!\right) \!\!\right) + O\!\left(\frac1n\right) \\
&\quad = - \cos \left( 2(n+1)\varphi - \frac{2n+1}{2} \sin(2\varphi) \right) + O\!\left(\frac1n\right) .
\ee
\end{widetext}
Here, in~1 we noted that $\left(2(n+1)\right)^{-1/2} = (2n)^{-1/2} \left( 1 + O(n^{-1})\right)$ and that $\frac{d}{d\varphi} \left(2\varphi - \sin(2\varphi)\right) = 4 \sin^2(\varphi)$; in~2, instead, we applied the sum-to-product formula $\sin\alpha - \sin \beta = 2\sin\left( \frac{\alpha-\beta}{2}\right) \cos\left( \frac{\alpha+\beta}{2}\right)$.

Now we can plug these estimates inside the integral~\eqref{homd_Fock}, obtaining that
\begin{widetext}
\bb
&\homd{\ketbra{n}}{\ketbra{n+1}} \\
&\qquad \geq \int_{-\sqrt{2n+1} \cos\epsilon}^{\sqrt{2n+1} \cos\epsilon} dx \left|\psi_n^2(x) - \psi_{n+1}^2(x)\right| \\
&\qquad = \int_{\epsilon}^{\pi-\epsilon} \left( \sqrt{2n+1}\sin\varphi\, d\varphi\right) \frac{1}{\pi} \sqrt{\frac{2}{n}}\left( \left| \cos \left( 2(n+1)\varphi - \frac{2n+1}{2} \sin(2\varphi) \right) \right| + O \!\left(\frac1n\right) \right) \\
&\qquad = \frac{2}{\pi} \int_{\epsilon}^{\pi-\epsilon} d\varphi\, \sin\varphi \left| \cos \left( 2(n+1)\varphi - \frac{2n+1}{2} \sin(2\varphi) \right) \right| + O \!\left(\frac1n\right) \\
&\qquad \texteq{3} \frac{4}{\pi^2} \left( (\pi-2\epsilon) \sin\epsilon - \int_\epsilon^{\pi-\epsilon} d\varphi\, (\varphi-\epsilon) \cos\varphi \right) + O \!\left(\frac1n\right) \\
&\qquad \texteq{4} \frac{8}{\pi^2} \cos\epsilon + O \!\left(\frac1n\right) .
\ee
\end{widetext}
Note that in~3 we integrated by parts, using Lemma~\ref{palla_lemma} to simplify the resulting expression. In~4, instead, we computed the remaining elementary integral, which is also done by parts.

Taking first the limit for $n\to \infty$ and then that for $\epsilon\to 0^+$ yields~\eqref{consecutive_Fock_L1}. To derive~\eqref{consecutive_Fock_hom}, just observe that if a sequence of strictly positive numbers $a_n>0$ satisfies that $\liminf_{n\to\infty} a_n >0$ then necessarily $\inf_n a_n >0$. In our case,
\bb
a_n\coloneqq&\ \left\|\ketbra{n} - \ketbra{n+1} \right\|_{\mathrm{hom}} \\
=&\ \left\|\psi_n^2 - \psi_{n+1}^2\right\|_{L^1} >0
\ee
for all $n$, because $\psi_n^2-\psi_{n+1}^2$ is a non-zero continuous function for all $n$.
\end{proof}

\bibliographystyle{apsrev4-1}
\bibliography{../../biblio}

\end{document}